\documentclass[11pt]{article}
\usepackage{float}
\usepackage{bm, commath}
\usepackage{natbib}
\usepackage{caption}
\usepackage{subcaption}
\usepackage{graphicx}
\usepackage{amsmath, amsfonts, amsthm}
\usepackage{float}
\usepackage{booktabs,siunitx}
\usepackage{mathabx}
\usepackage{url}
\usepackage{multirow}
\usepackage{amssymb}
\usepackage{blkarray}
\usepackage{bbm}
\usepackage{multicol}
\usepackage{authblk}
\usepackage[shortlabels]{enumitem}
\usepackage[flushleft]{threeparttable}
\usepackage{enumitem}
\usepackage{thmtools}
\usepackage{mathtools}

\newcommand{\albert}{\textcolor{black}}

\newcommand{\E}{\mathbb{E}}

\usepackage{bm, commath}
\usepackage{comment} 
\usepackage{mathrsfs}
\usepackage{appendix}  

\usepackage[left=1in,right=1in,bottom=0.9in,top=0.9in]{geometry}

\usepackage{algorithm}
\usepackage{algpseudocode}
\usepackage{mathtools} 
\usepackage{caption}
\usepackage{bbm}
\usepackage{textcomp}

\usepackage{authblk}

\usepackage{sectsty}

\newtheorem{theorem}{Theorem}

\newtheorem{remark}{Remark}

\newtheorem{lemma}{Lemma}

\usepackage[doublespacing]{setspace}
\usepackage[colorlinks=true,
            linkcolor=blue,
            citecolor=blue,
            urlcolor=blue]{hyperref}

\begin{document}

\def\spacingset#1{\renewcommand{\baselinestretch}%
{#1}\small\normalsize} \spacingset{1.7}


\sectionfont{\bfseries\large\sffamily}%

\subsectionfont{\bfseries\sffamily\normalsize}%

\title{Developing Performance-Guaranteed Biomarker Combination Rules with Integrated External Information under Practical Constraint}

\author[1]{Albert Osom\thanks{Email: \texttt{aosom@uw.edu}}}
\author[2]{Camden Lopez}
\author[3]{Ashley Alexander}
\author[4]{Suresh Chari}
\author[2]{Ziding Feng}
\author[2,1]{Ying-Qi Zhao}

\affil[1]{Department of Biostatistics, University of Washington, Seattle, WA}
\affil[2]{Public Health Sciences Division, Fred Hutchinson Cancer Center, Seattle, WA}
\affil[3]{Kelsey Research Foundation, Houston, TX}
\affil[4]{Department of Gastroenterology, University of Texas MD Anderson Cancer Center, Houston, TX}
\date{}
\maketitle
\vspace{-35pt}
\bigskip
\begin{abstract}
In clinical practice, there is significant interest in integrating novel biomarkers with existing clinical data to construct interpretable and robust decision rules. Motivated by the need to improve decision-making for early disease detection, we propose a framework for developing an optimal biomarker-based clinical decision rule that is both clinically meaningful and practically feasible. Specifically, our procedure constructs a linear decision rule designed to achieve optimal performance among class of linear rules by maximizing the true positive rate while adhering to a pre-specified positive predictive value constraint. Additionally, our method can adaptively incorporate individual risk information from external source to enhance performance when such information is beneficial. We establish the asymptotic properties of our proposed estimator and compare to the standard approach used in practice through extensive simulation studies. Results indicate that our approach offers strong finite-sample performance. We also apply the proposed methods to develop biomarker-based screening rules for  pancreatic ductal adenocarcinoma (PDAC) among new-onset diabetes (NOD) patients.
\end{abstract}

\section{Introduction} \label{intro}

Early detection is critical in the fight against cancer and other diseases, as it offers the potential to identify diseases at a stage where they are most treatable. The concept of early detection is grounded in the principle that earlier intervention following the detection of disease can lead to better outcomes, including higher survival rates, reduced treatment complexity, and improved quality of life for patients. 
An immense amount of research has been dedicated to identifying candidate biomarkers to aid in early detection. 
The primary challenge lies in effectively combining these biomarkers to develop a clinical decision rule that satisfies specific practical and clinical needs, while optimizing the detection of the target disease. 
In this regard, \cite{vickers2006decision} and \cite{kerr2016assessing} proposed net benefit and standardized net benefit as a measure of clinical utility of risk models, which balances the trade-off between the harm and benefit of the  model at a clinically relevant cost-benefit ratio.  
Recently, \cite{wang2020learning} developed a learning-based biomarker-assisted rule that optimizes clinical benefit under a pre-specified risk constraint, offering a flexible framework for developing both linear and nonlinear rules.  \cite{meisner2021combining} introduced a framework for linearly combining biomarkers that maximizes the true positive rate (TPR) for a fixed false positive rate (FPR).

Our work is motivated by early detection of pancreatic ductal adenocarcinoma (PDAC). PDAC has an overall 5-year survival rate of around $13\%$ \citep{siegel2025cancer} and has been projected to be the second highest cancer-related cause of death in the United States by 2030 \citep{rahib2014projecting}. 
Due to the low prevalence of PDAC, biomarker-based screening must target high-risk populations to be feasible \textemdash \ such as patients over 50 years with new-onset hyperglycemia and diabetes (NOD) \citep{chari2005probability}. Further, cost-effective screening tests are essential to selectively identify PDAC cases, as it is impractical to recommend all at-risk individuals for further workup that is more expensive or invasive. For instance, over one million Americans are diagnosed with new-onset diabetes each year, making broad screening in this population prohibitively expensive. \albert{ In the setting of population level screening, clinician are often interested in targeting a controllable work-up burden from recommending further diagnostic procedure, especially when such follow-up evaluations are costly, invasive, or limited by resource constraints. Positive predictive value (PPV) naturally serves as a clinically meaningful criterion for such need, because it directly quantifies the expected number needed to screen (NNS) to detect one cancer case. For example, in ovarian cancer screening, clinical guidelines recommend a minimum PPV of $10\%$ to justify invasive diagnostic surgery \citep{jacobs2004progress}. A rule achieving PPV of $10\%$ is expected to recommend 10 patients for further diagnostic work-up to identify one cancer case. This NNS interpretation is intuitive and compelling for clinicians when evaluating the feasibility of a screening modality. In contrast, existing biomarker-based decision rules that maximize TPR subject to an FPR constraint \citep{wang2020learning,meisner2021combining} are not well suited for this setting. FPR is a prevalence-free measure and does not directly determine the clinical work-up burden of a screening rule. It is not straightforward to translate an FPR-constrained optimization into a rule that satisfies a pre-specified PPV requirement. For a fixed disease prevalence, PPV is a nonlinear function of both TPR and FPR, and the mapping from PPV to FPR is not identifiable, as there exist infinitely many (TPR, FPR) pairs that achieve the same PPV.  Furthermore, in our setting, TPR is both the quantity being optimized and a component of the PPV constraint itself. Consequently, for a fixed PPV level, the admissible FPR is not unique but depends on the achieved TPR. Fixing an FPR threshold to approximate a target PPV while maximizing TPR can therefore yield a suboptimal rule relative to an approach that directly optimizes the clinical objective. See Table~\ref{illustration} for an illustrative numerical example.}

Additionally, when developing a new  biomarker-based decision rule, there are often published risk models/algorithms available that can provide valuable insights into patient risk profile.  For instance, the Enriching New-Onset Diabetes for Pancreatic Cancer (ENDPAC) model was developed and validated for PDAC screening in NOD patients, using patient characteristics such as weight change, blood glucose change, and age at diabetes onset \citep{sharma2018model}. Other examples include the MyProstateScore test for prostate cancer, endorsed by NCCN guidelines for prebiopsy risk stratification \citep{NCCNprostate}, and the Framingham Risk Score for coronary heart disease \citep{wilson1998prediction}. Integrating this information could enhance the discriminatory power of a new biomarker decision rule.

\begin{table}[t]
\caption{Comparison of the performance of the FPR-constrained method of \cite{meisner2021combining}, which maximizes TPR subject to an FPR constraint, with our proposed method (DOOLR), which directly maximizes TPR subject to a PPV constraint. To align the FPR constraint in \cite{meisner2021combining} with a given PPV target, we solve the optimization problem over a range of FPR constraints and select the value that yields the desired PPV level. Data is generated under a linear decision rule; see Section~\ref{cohort-linear}. We report average TPR and PPV with corresponding standard errors (in parentheses) in the test sample across 500 simulation replicates for PPV constraints $\alpha \in \{0.030, 0.040, 0.045\}$ and prevalence $p_1 = 0.01$. Across all values of $\alpha$, DOOLR achieves higher TPR than the FPR-constrained method of \cite{meisner2021combining}.} 
\label{illustration}
\centering
\scalebox{0.8}
{\begin{tabular}{ccccc}
\toprule
\multicolumn{1}{c}{$\alpha$} & \multicolumn{1}{c}{\textbf{Measure}}& \multicolumn{1}{c}{\textbf{Meisner et al. (2021)}}& \multicolumn{1}{c}{\textbf{DOOLR}}  \\
\midrule
0.030 & TPR & 0.870(0.177)  & 0.942(0.139)    \\
   & PPV &  0.030(0.001) & 0.035(0.008)   \\\\
0.040 & TPR &  0.818(0.198) & 0.926(0.120)    \\
   & PPV & 0.040(0.001)  &  0.044(0.007)  \\\\
0.045 & TPR &  0.796(0.187) & 0.907(0.153)    \\
   & PPV & 0.045(0.002)  &  0.048(0.013)  \\
\bottomrule
\end{tabular}}
\end{table}

In the context of risk prediction, various strategies have been proposed to leverage external information when developing a new risk model. Assuming there is a published risk prediction model available, we consider two established approaches discussed in the literature. The first estimates the parameters of the new risk prediction model using penalized likelihood, incorporating a Kullback–Leibler divergence penalty to enforce similarity between the new and the published model \citep{kullback1951information,hector2024turning,wang2021kullback,wang2023incorporating}. The second employs a constraint-based maximum likelihood estimation framework, restricting the extent to which the parameters of the new prediction model deviate from that of the published model to encourage alignment \citep{cheng2019informing,han2023integrating,chatterjee2016constrained}. Both approaches assume either there is access to an external data or correct specification of the published risk model. In contrast, our approach imposes no such assumptions, relying solely on the availability of external risk information \textemdash\ such as calculated risk scores under a published risk model, algorithm, or risk calculator \textemdash\ while maintaining flexibility in model specification. Moreover, as these existing methodologies primarily focus on risk prediction, they do not directly address our research objective, rendering them less suitable for our framework.

\albert{In this paper, we introduce a new criterion, termed PPV-constrained benefit function, from which we develop biomarker-based decision rules. Under this criterion, we learn decision rules that maximize the benefit of screening measured by TPR while enforcing clinical usefulness through a PPV constraint}. We derive the theoretically optimal rule, based upon which a plug-in rule could be formed using either parametric or flexible nonparametric estimators, such as random forests \citep{breiman2001random}, additive models \citep{hastie1990generalized}, or an ensemble of estimators implemented via the SuperLearner \citep{van2007super}. Considering the importance of interpretability in clinical practice, we further develop an approach using Direct Optimization for Optimal Linear decision Rule (DOOLR). In cases where the optimal decision rule is not linear, the proposed method ensures the development of a linear decision rule that delivers the best possible performance within the entire class of linear rules. Moreover, when external risk information, such as decisions from a published risk model or established biomarkers is available, we extend the linear decision rule to incorporate this auxiliary information, thereby improving predictive performance while preserving interpretability. We formalize this integration by introducing a penalty term into the objective function that adaptively regulates the extent to which external information is incorporated. Specifically, the penalty term quantifies the cost associated with misalignment between the decisions generated by the new biomarker rule and those recommended by the published risk model among patients with positive disease status. \albert{Importantly, our external information transfer approach does not require access to individual-level external risk scores; it only requires binary decisions (i.e., whether further screening is recommended) derived from an external risk model.}

The manuscript is organized as follows: we introduce the PPV-constrained benefit function function in Section \ref{methodology}. We propose algorithms for estimating the optimal decision rule that maximizes the PPV-constrained benefit function in Section \ref{estimation}. In Section \ref{asymptotic} , we discuss the asymptotic properties of the proposed estimator of the optimal linear rule. Simulation studies evaluating the proposed procedures are presented in Section \ref{simulation}. In Section \ref{application}, we illustrate our method to Pancreatic Cancer and a retrospective cohort of controls enriched for new onset Diabetes and pre-diabetes for Research Analyses (PANDORA) study. Our goal is develop biomarker-assisted decision rule for screening of PDAC among new-onset diabetes (NOD) patients. In Section \ref{discussion}, we provide a discussion. 

\section{ Methodology\label{methodology}}

\subsection{PPV-constrained benefit function}
Let $D$ denote a binary outcome, with $D = 1$ indicating presence of the disease of interest and $0$ otherwise. 
Let ${\bm X}$ be the set of biomarkers and other clinical characteristics used for deriving clinical decision rules. Denote by $\mathscr{R}=\{d(\bm X) : d(\bm X)\in \{0,1 \} \}$ the class of individualized clinical decision rules, such that $d(\bm X)=1$ means patient is recommended for further diagnostic procedure, and  $d(\bm X)=0$ otherwise. Let $\text{TPR}(d(\bm X))=\text{pr}(d(\bm X)=1|D=1)$, and $\text{PPV}(d(\bm X))=\text{pr}(D=1 | d(\bm X)=1)$. We are interested in constructing biomarker combination rules that maximizes TPR under a pre-specified clinically meaningful positive predictive value, $\alpha$. The problem can be formulated as  
\begin{align}
   \underset{d(\bm X) \in \mathscr{R}}{ \text{max }} \text{TPR}(d(\bm X)) \textit{ subject to } \text{ PPV}(d(\bm X)) \ge \alpha.  \label{eqn:1}
\end{align}
For $k \in \{0,1 \}$, let $p_k \coloneq \text{pr}(D = k)$, where $p_1$ denote the disease prevalence, and $\text{pr}(\bm X)$ and $\text{pr}(D = k|\bm X)$ denote the density function of $\bm X$ and the conditional probability of $D = k$ given $\bm X$, respectively. We assume $p_1$ is either known or can be estimated from independent data sources. We further assume that both $\text{pr}(D = 1|\bm X)$ and $\text{pr}(\bm X)$ are continuous functions of $\bm X$, and $\delta_0 < \text{pr}(D = k|\bm X) < 1 - \delta_0$ and $\delta_1 < p_k < 1 - \delta_1$ almost surely for some positive constants $\delta_0$ and $\delta_1$. For notational simplicity, let $\gamma=p_1/(1-p_1)$, $\eta_1(\bm X) = \text{pr}(D = 1|\bm X)/p_1$ and $\eta_0(\bm X) = \text{pr}(D = 0|\bm X)/p_0$, where $p_0=1-p_1$. We also denote $\mathbb{E}_{\bm X}$ as the expectation with respect to the marginal distribution of $\bm X$.
We can then write \eqref{eqn:1} as 
\begin{eqnarray}
      \underset{d(\bm X) \in \mathscr{R}}{ \text{max }} && \mathbb{E}_{\bm X}[\mathbbm{1}\{d(\bm X)=1\}\eta_1(\bm X)] \nonumber\\
    \textit{ subject to } && \frac{\gamma \mathbb{E}_{\bm X}[\mathbbm{1}\{d(\bm X)=1\}\eta_1(\bm X)]}{\gamma \mathbb{E}_{\bm X}[\mathbbm{1}\{d(\bm X)=1\}\eta_1(\bm X)] + \mathbb{E}_{\bm X}[\mathbbm{1}\{d(\bm X)=1\}\eta_0(\bm X)] } \ge \alpha. \label{eqn:2}
\end{eqnarray}
Our goal is to find the optimal decision rule $d^*(\bm X)$ that solves \eqref{eqn:2}. The Lagrangian corresponding to \eqref{eqn:2} is given as 
\begin{align*}
    L(d(\bm X),\lambda)=  \quad \E[\mathbbm{1}\{d(\bm X)=1\} \{\eta_{1}(\bm X)-\lambda\alpha\gamma\eta_1(\bm X) -\lambda\alpha\eta_0(\bm X)+\lambda\gamma\eta_1(\bm X) \}].
\end{align*}
\albert{Let $d^{*}(\bm X)$ denote the optimal rule that solves the optimization problem in \eqref{eqn:2}. We can obtain $d^{*}(\bm X)$ by maximizing $L(d(\bm X),\lambda)$ for a fixed $\lambda$ and choose the optimal $\lambda$ as the solution to  $\mathbb{E}_{\bm X}[\mathbbm{1}\{d(\bm X)=1\} (\gamma\eta_1(\bm X))-\alpha\gamma\eta_1(\bm X) -\alpha\eta_0(\bm X)]= 0$. The optimal rule is derived as  $d^{*}(\bm X)= \mathbbm{1} \{f_{\lambda^*} (\bm X) >0 \}$, where $f_{\lambda^*}(\bm X) = \eta_{1}(\bm X)-\lambda^*\alpha\gamma\eta_1(\bm X) -\lambda^*\alpha\eta_0(\bm X)+\lambda^*\gamma\eta_1(\bm X)$ for a pre-specified $\alpha$, and $\lambda^* >0$ solves 
$\mathbb{E}_{\bm X}[\mathbbm{1}\{f_{\lambda}(\bm X)>0\} (\alpha\gamma\eta_1(\bm X) +\alpha\eta_0(\bm X)-\gamma\eta_1(\bm X))]= 0.$ We show optimality of $d^{*}(\bm X)$ in Supplementary Appendix~A.}

\subsection{Estimation of the optimal decision rule}\label{estimation}

Suppose $\{D_i,\bm{X}_i\}_{i=1}^{n}$ data are collected from $n$ independent identically distributed subjects. We denote $n_1$ by the number of disease subjects and $n_0$ the number of non-disease subjects such that $n=n_0 + n_1$. Given the form of the theoretical optimal decision rule, a straightforward method for estimation is to replace the expectations in the theoretical optimal decision rule and the constraint with their empirical analogs. Let 
$\hat{\eta}_1(\bm X)=\hat{\text{pr}}(D=1|\bm X)/p_1$ and $\hat{\eta}_0(\bm X)=\hat{\text{pr}}(D=0|\bm X)/p_0$ be estimates of $\eta_1(\bm X)$ and $\eta_0(\bm X)$ respectively, where $\text{pr}(D=1|\bm X)$ can be estimated using either parametric (e.g., logistic regression) or data adaptive approaches (e.g., random forest, neural network, \texttt{SuperLearner}, etc).
Subsequently, the estimated optimal decision rule is 
$\hat{d}_{\hat{\lambda}}(\bm X)= \mathbbm{1} \{\hat{\eta}_{1}(\bm X)-\hat{\lambda}\alpha\gamma\hat{\eta}_1(\bm X) -\hat{\lambda}\alpha\hat{\eta}_0(\bm X)+\hat{\lambda}\gamma\hat{\eta}_1(\bm X) >0 \},$
where $\hat{\lambda}$ solves
\begin{align*}
    n^{-1}\sum_{i=1}^n[\mathbbm{1}\{\hat{d}_{\hat{\lambda}}(\bm X_i)=1\} (\alpha\gamma\hat{\eta}_1(\bm X_i) +\alpha\hat{\eta}_0(\bm X_i)-\gamma\hat{\eta}_1(\bm X_i))]= 0.
\end{align*}
We refer to the method discussed above as the plug-in approach. This method is flexible and can accommodate various estimation techniques (e.g., random forest, neural network, additive models, local polynomials, etc). 

\begin{remark}
    When estimating $\text{pr}(D=1|\bm{X})$ in the above plug-in rule using data under case-control sampling, we need to make adjustments to accurately estimate $\text{pr}(D=1|\bm{X})$. To be specific, let $S$ be the indicator of being included into the case-control sample. We can  obtain an estimate of $\text{pr}(D=1|\bm X,S)$ under the case-control data. However, we are typically more interested in estimating $\text{pr}(D=1|\bm X)$. In this case, we make an adjustment using the 
     relationship $$\frac{\hat{\text{pr}}(D=1|\bm X)}{\hat{\text{pr}}(D=0|\bm X)}= \frac{\hat{\text{pr}}(D=1|\bm X,S)}{\hat{\text{pr}}(D=0|\bm X,S)}\frac{n_0}{n_1} \frac{p_1}{1-p_1}$$ \citep{huang2010assessing}.  If we estimate $\text{pr}(D=1|\bm{X})$ using logistic regression, this adjustment simplifies to adding $\text{log}[{p_1n_0}/\{(1-p_1)n_1\}
 ]$ to the intercept of the estimated coefficient.
\end{remark}
In clinical practice, simpler decision rules are preferred for their ease of implementation and interpretability. A common choice is linear decision rules \textemdash\ this is usually obtained in two steps; first estimate a patient’s disease probability using logistic regression and apply a risk score cutoff to achieve the pre-specified clinical constraint. While intuitive, this method can underperform when model assumptions are incorrect and is often ad hoc, meaning it may not yield the optimal rule for a given performance measure. This challenge led us to develop the Direct Optimization for Optimal Linear Decision Rule (DOOLR) approach, which constructs a robust linear rule that is asymptotically optimal within the class of all linear rules.

\subsubsection{Direct optimization for optimal linear decision rule} \label{DOOLR}
Let $\mathscr{R}_L$ be the class of linear decision rules of the form $d(\bm X,\bm \beta)=\mathbbm{1}\{\mathscr{X}^T\bm{\beta} >0 \}$ where $\mathscr{X}=(1,\bm X)$ for $p$-dimensional feature variables $\bm{X}$ and $\bm \beta=(\beta_0,\bm \beta_1)$. Denote by $\mathscr{X}_{1i}$ and $\mathscr{X}_{0j}$ the covariates for diseased and non-diseased patients, respectively.  Let $d_L^*(\bm{X},\bm \beta^{*})=\mathbbm{1}\{\mathscr{X}^T\bm{\beta^{*}} >0 \}$ be the optimal linear rule that maximizes the objective in \eqref{eqn:1} among $\mathscr{R}_L$. The empirical analog to \eqref{eqn:1} using the observed data for the class of linear decision rules can be written as
\begin{eqnarray}
  &&\underset{\bm \beta \in\mathbb{R}^{p+1}}{ \text{max }}  n_1^{-1}\sum_{i=1}^{n_1}[\mathbbm{1}\{\mathscr{X}_{1i}^T\bm{\beta} >0\}] \nonumber\\
    \textit{ subject to } && \widehat{\text{PPV}}(\bm \beta)=\frac{\gamma n_1^{-1}\sum_{i=1}^{n_1}[\mathbbm{1}\{\mathscr{X}_{1i}^T\bm{\beta} >0\}]}{\gamma n^{-1}_1\sum_{i=1}^{n_1}[\mathbbm{1}\{\mathscr{X}_{1i}^T\bm{\beta} >0\}] + n^{-1}_0\sum_{j=1}^{n_0}[\mathbbm{1}\{\mathscr{X}_{0j}^T\bm{\beta} >0\}] } \ge \alpha. \label{eqn:3}
\end{eqnarray}
\albert{The optimization problem in \eqref{eqn:3} involving the zero–one indicator is known to be nondeterministic polynomial-time hard (NP-hard) \citep{natarajan1995sparse}. The presence of the indicator function introduces discontinuity into the objective and constraint functions, rendering gradient-based optimization methods infeasible. Our goal is not to circumvent this fundamental computational barrier, but to develop a practically tractable approximation that targets the optimal decision rule within the class of linear rules. To address this computational intractability, approximating the zero-one indicator with a smooth surrogate function is a widely used approach in the classification literature, demonstrating both computational and theoretical advantages \citep{bartlett2006convexity}. While commonly used convex surrogate losses (e.g., logistic or hinge loss), are computationally convenient, they do not in general preserve optimality within a restricted model class and may converge to suboptimal linear rules under model misspecification.}
\begin{figure}[t]
    \centering
    \includegraphics[width=12cm]{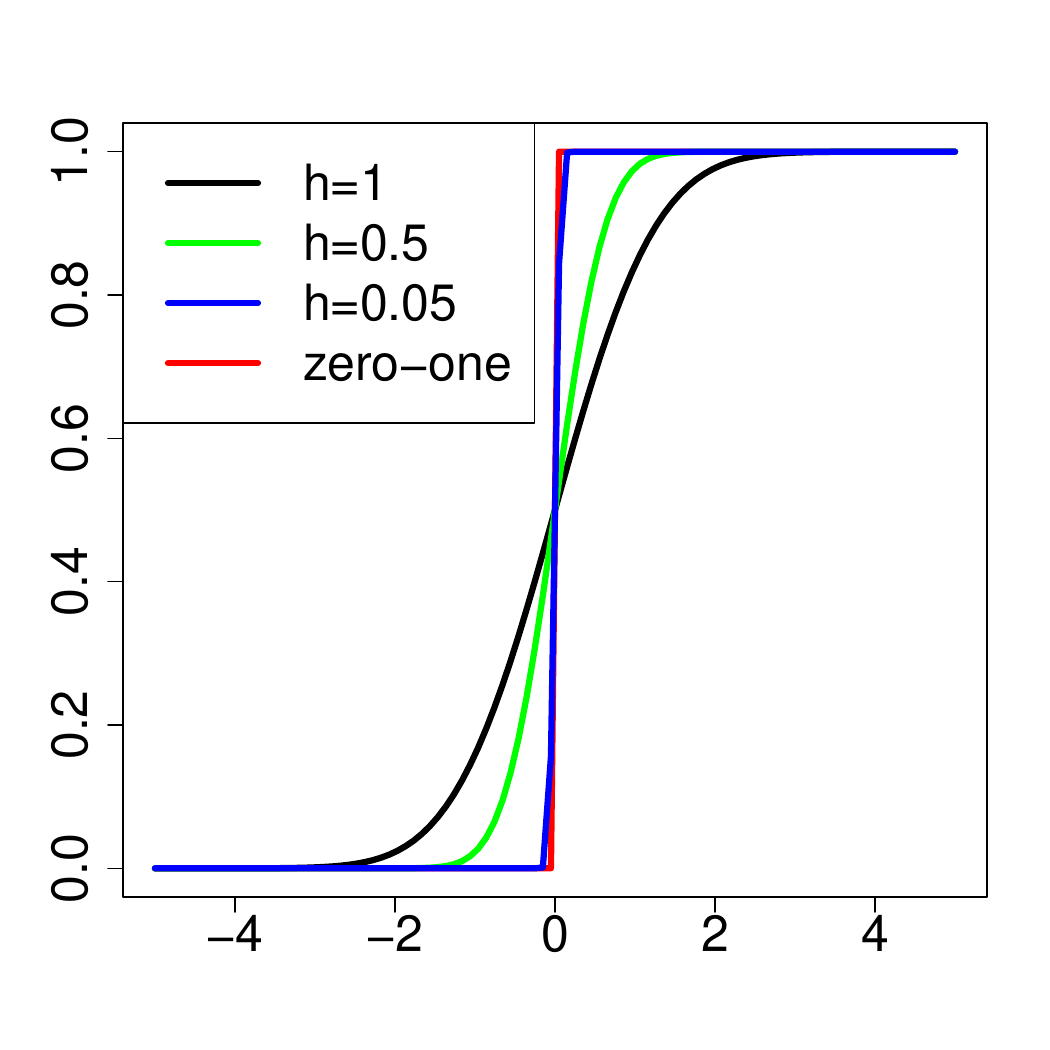}
    \caption{Comparison of zero-one indicator with $\Phi{(x/h)}$ for different values of $h$. We can observe that, when $h$ approaches zero, the approximation of the zero-one indicator with $\Phi{(x/h)}$ becomes more accurate. }
     \label{fig:1}
\end{figure}
To ensure we estimate the optimal linear decision rule with guaranteed performance, we also propose using a smooth approximation loss function that converges to the zero-one indicator. Specifically, we choose our surrogate function to be $\Phi(x/h)$, where $\Phi$ is the standard normal distribution function \citep{ma2007combining, fong2016combining, lin2011selection}.  As $h$ approaches zero, the surrogate function converges to the indicator function (see Figure \ref{fig:1}). \albert{The parameter $h$ in $\Phi(x/h)$ controls how closely the surrogate function approximates the indicator function and should be chosen in a data-adaptive manner. Ideally, $h$ would be selected through cross-validation to ensure that it yields the highest-performing rule. However, cross-validation is  computationally burdensome as each candidate value of $h$ requires solving 
a nonlinear optimization problem. For this reason, we adopt a practical choice
$ h = n^{-1/3}\,\mathrm{StdDev}\!\left(\mathcal{X}^{T}\hat{\bm \beta}^{0}/\lVert 
\hat{\bm \beta}^{0}\rVert\right)$, where $\hat{\bm \beta}^{0}$ is an initial estimate of the coefficients via say logistic regression. This choice of $h$ has demonstrated strong empirical performance in prior work \citep{lin2011selection,meisner2021combining}}. Replacing the indicator function in \eqref{eqn:3} with the proposed surrogate function, the relaxed optimization problem is given as
\begin{eqnarray}
  &&\underset{\bm \beta \in \mathbb{R}^{p+1}}{ \text{max }}  n_1^{-1}\sum_{i=1}^{n_{1}}\Phi\left(\frac{\mathscr{X}_{1i}^T\bm{\beta}}{h}\right)  \nonumber\\
   \textit{ subject to }   && \frac{\gamma n_1^{-1}\sum_{i=1}^{n_1}\Phi\left(\frac{\mathscr{X}_{1i}^T\bm{\beta}}{h}\right)}{\gamma n^{-1}_1\sum_{i=1}^{n_1}\Phi\left(\frac{\mathscr{X}_{1i}^T\bm{\beta}}{h}\right) + n^{-1}_0\sum_{j=1}^{n_0}\Phi\left(\frac{\mathscr{X}_{0j}^T\bm{\beta}}{h}\right) } \ge \alpha. \quad \label{eqn:4}
\end{eqnarray}
The solution to the relaxed problem in \eqref{eqn:4} $\hat{\bm \beta}^{\Phi_{h}}$, maximizes the Lagrangian 
\begin{align*}
L_{n}^{\Phi}(\bm \beta,\lambda) & =  n_1^{-1}\sum_{i=1}^{n_{1}}\Phi\left(\frac{\mathscr{X}_{1i}^T\bm{\beta}}{h}\right) + \hat{\lambda}^{\Phi}\gamma n_{1}^{-1}\sum_{i=1}^{n_1}\Phi\left(\frac{\mathscr{X}_{1i}^T\bm{\beta}}{h}\right) -\hat{\lambda}^{\Phi} \alpha\gamma n_1^{-1}\sum_{i=1}^{n_1}\Phi\left(\frac{\mathscr{X}_{1i}^T\bm{\beta}}{h}\right)\\
 &- \hat{\lambda}^{\Phi} \alpha n^{-1}_0\sum_{j=1}^{n_0}\Phi\left(\frac{\mathscr{X}_{0j}^T\bm{\beta}}{h}\right),  
\end{align*} 
where $\hat{\lambda}^{\Phi}$ solves $\widehat{\text{PPV}}(\mathbbm{1}\{\mathscr{X}^T\hat{\bm{\beta}}^{\Phi_{h}}>0\} )-\alpha=0$. To simplify notation, we write $\hat{\bm \beta}^{\Phi}$ to denote  $\hat{\bm \beta}^{\Phi_{h}}$ while keeping in mind that the estimator implicitly depends on $h$. \albert{This is a smooth objective function which we solve using an off-the-shelf non-linear optimization package in R (\texttt{nlm}). We then perform a grid search for the Lagrange multiplier to identify the value that satisfies the PPV constraint, $\widehat{\text{PPV}}(\mathbbm{1}\{\mathscr{X}^T\hat{\bm{\beta}}^{\Phi}>0\} )-\alpha=0$. Because the decision rule depends only on the sign of  $\mathscr{X}^T\hat{\bm{\beta}}$, we enforce the $l_2$ constraint such that $\|\hat{\bm{\beta}}^{\Phi} \|^{2}_{2}=1$, which does not affect the resulting decision rule. In Supplementary Appendix B, we summarize the full procedure in Algorithm 1 and  justification of the optimization steps.}

We denote the estimated linear decision function obtained by the above procedure as $\hat d_L(\bm X,
\hat{\bm{\beta}}^{\Phi})=\mathbbm{1}\{\mathscr{X}^T\hat{\bm{\beta}}^{\Phi}>0\}$. We will show in Section \ref{asymptotic} that $\hat d_L(\bm X,\hat{\bm \beta}^{\Phi})$ converges to the true best linear rule, $d_L^*(\bm X,
\bm \beta^*)$, asymptotically.

\subsubsection{Incorporating existing information to improve linear combination rule}
In many situations, there already exist risk models or scores developed to aid in patient risk stratification. If these existing risk information are useful, it would be beneficial to include them when learning the new biomarker combination rule. In this section, we extend the DOOLR approach to allow incorporation of useful external risk information.

\albert{Assume there exists an external risk algorithm or published model for which only the resulting binary decisions (i.e., recommend further screening or not) are available for all patients in the current dataset. Our framework does not require access to the original risk scores or model parameters; it only requires the external decisions. In some settings, the underlying risk scores may also be available. Let $\bm Z \subseteq \bm X$ denote the subset of features used to develop the external model, and let $r(\bm Z)$ denote the associated risk score with clinical threshold $\delta_0$, where larger values correspond to higher disease risk. For simplicity, we assume the external decision rule takes the threshold form $\mathbbm{1}\{r(\bm Z)-\delta_0>0\}$. For example, if the external model is logistic regression, then $r(\bm Z) = \bm Z^\top \widetilde{\bm\beta}$, where $\widetilde{\bm\beta}$ are published coefficients of the the model. More generally, $r(\bm Z)$ may arise from a machine learning or domain-informed algorithm. Our method, however, relies only on the induced binary decisions and remains applicable even when the underlying scores are unavailable.} 

\albert{To efficiently incorporate external risk information when developing the new biomarker combination rule, we propose a strategy that transfers external information only when it improves the performance metric. Since our goal is to maximize TPR, we introduce a penalty term in the objective function in \eqref{eqn:4} that quantifies disagreement between the decision rules among disease patients. The penalty term is}
\begin{align*}
   \frac{1}{n_1}\sum_{i=1}^{n_1} \mathbbm{1}\{ \mathscr{X}_{1i}^T\bm{\beta}\cdot (r(\bm{Z}_i)-\delta_0) <0 \}. 
\end{align*}
When the two rules align, both $\mathscr{X}_{1i}^T\bm{\beta}$ and $(r(\bm{Z}_i)-\delta_0)$ have the same sign. The penalty therefore imposes an additional cost when the new rule disagrees with the external rule on true disease cases. Additionally, we introduce a tuning parameter, $\eta$, to control the extent of information transfer from the external risk model. This ensures that if the external model is not informative, its influence on developing the biomarker combination rule remains minimal.  

We now propose a smooth Information Transfer (IT)-based penalized maximization to construct the linear decision rule, built upon \eqref{eqn:4}, we term this approach IT-DOOLR. The optimization problem is framed as
\begin{align}
   & \underset{\bm{\beta} \in \mathbb{R}^{p+1} }{ \text{max }}  n_1^{-1}\sum_{i=1}^{n_1} \Phi\left(\frac{\mathscr{X}_{1i}^T\bm{\beta}}{h}\right) - \eta   n_{1}^{-1}\sum_{i=1}^{n_1} \Phi \left(\frac{-\mathscr{X}_{1i}^T\bm{\beta}\cdot (r(\bm{Z}_i)-\delta_0) }{h}\right)\nonumber \\
   \textit{ subject to } &
\frac{\gamma n_1^{-1}\sum_{i=1}^{n_1}\Phi\left(\frac{\mathscr{X}_{1i}^T\bm{\beta}}{h}\right)}{\gamma n_{1}^{-1} \sum_{i=1}^{n_{1}}\Phi\left(\frac{\mathscr{X}_{1i}^T\bm{\beta}}{h}\right)+  n_0^{-1}\sum_{j=1}^{n_0}\Phi\left(\frac{\mathscr{X}_{0j}^T\bm{\beta}}{h}\right) } \ge \alpha. 
\end{align} 
Similarly, we can solve for $\hat{\bm \beta}_{\text{IT}}^{\Phi} =\underset{\bm \beta \in \mathbb{R}^{1+p}}{\text{argmax }}L_{\text{IT},n}^{\Phi}(\bm \beta,\lambda)$ using gradient-based methods, where  
\begin{align*}
L_{\text{IT},n}^{\Phi}(\bm \beta,\lambda)& =  n_1^{-1}\sum_{i=1}^{n_{1}}\Phi\left(\frac{\mathscr{X}_{1i}^T\bm{\beta}}{h}\right) + \lambda\gamma n_{1}^{-1}\sum_{i=1}^{n_1}\Phi\left(\frac{\mathscr{X}_{1i}^T\bm{\beta}}{h}\right)
- \lambda \alpha\gamma n_1^{-1}\sum_{i=1}^{n_1}\Phi\left(\frac{\mathscr{X}_{1i}^T\bm{\beta}}{h}\right) \nonumber\\
& -\lambda \alpha n^{-1}_0\sum_{j=1}^{n_0}\Phi\left(\frac{\mathscr{X}_{0j}^T\bm{\beta}}{h}\right) - \eta  n_{1}^{-1}\sum_{i=1}^{n_1} \Phi \left(\frac{-\mathscr{X}_{1i}^T\bm{\beta}\cdot (r(\bm{Z}_i)-\delta_0) }{h}\right),
\end{align*} 
and $\hat{\lambda}^{\Phi}_{\text{IT}}$ solves $\widehat{\text{PPV}}(\mathbbm{1}\{\mathscr{X}^T\hat{\bm{\beta}}^{\Phi}_{\text{IT}}>0\} )-\alpha=0$. 
We will select  $\eta$ through cross-validation to ensure efficient information transfer. The procedure will leverage external risk information when it enhances the TPR and maintains control over the  PPV by selecting larger values of $\eta$, while minimizing negative information transfer by selecting $\eta$ values closer to zero.

\section{ Theoretical properties\label{asymptotic}}

In this section, we discuss the asymptotic properties of the estimated linear rule using the proposed DOOLR approach. We present asymptotic guarantee that the best linear rule is achieved. We further show that the TPR under the estimated linear rule converges to the optimal TPR  in the linear class and PPV is controlled for the pre-specified $\alpha$. 
We assume the following standard conditions.
\begin{enumerate}
    \item[(A1)] Suppose $n_0/n, n_1/n \in (0,1)$  where $n=n_1+n_0 \rightarrow \infty$.
    \item[(A2)] For each $d\in \{0,1 \}$, observations $\bm X_{di}, i = 1, 2, … , n_d$, are independent and identically distributed $p$-dimensional random vectors with distribution function $F_d$.
   \albert{\item[(A3)] For each $d\in \{0,1\}$, the conditional distribution of 
$\bm X $ given $D=d$ is not supported on any proper linear subspace of 
$\mathbb{R}^{p}$; that is, for all proper linear subspaces 
$S \subset \mathbb{R}^{p}$, $\text{Pr}(\bm X \in S|D=d)<1$.}
    \item[(A4)] For each $d \in \{0, 1\}$, the distribution and quantile functions of $\bm X^T\bm \beta_1$ given $D = d$ are globally Lipschitz continuous uniformly over $\bm \beta_1 \in \mathbb{R}^{p}$ such that $||\bm \beta_1||=1$.
    \item[(A5)] The map $(\beta_0,\bm \beta_1) \mapsto \text{TPR}(\beta_0,\bm \beta_1)$ is globally Lipschitz continuous over $\Omega=\{\bm \beta_1 \in \mathbb{R}^{p}, \beta_0 \in \mathbb{R}: ||\bm \beta_1||=1\}$.
\end{enumerate}

 To summarize, Conditions (A1) and (A2) are standard assumptions in asymptotic analysis. Condition (A3) ensures $\bm X$ has full support in $\mathbb{R}^{p}$ and cannot be restricted to a lower-dimensional subspace for some $D$. Conditions (A4) and (A5) are less stringent smoothness conditions used to show asymptotic results. The smoothness condition ensures controlled variations in the functions.  The above conditions are similar to those stated in \cite{meisner2021combining}.

\begin{theorem}
     Under conditions (A1)–(A5), we have that as $h \rightarrow 0$  
     \begin{enumerate}[a)]
          \item $TPR(\hat{\bm \beta}^{\Phi}) \rightarrow TPR(\bm \beta^*)$ in probability,
          \item  $\text{lim inf }  PPV(\hat{\bm \beta}^{\Phi})\ge \alpha$ almost surely,
          \item $\hat{\bm \beta}^{\Phi} \rightarrow \bm \beta^*$ in probability.
 \end{enumerate}  \label {theorem1}
\end{theorem}

 The proof of the above theorem is based on  Lemma 2 from \cite{meisner2021combining} and M-estimation theory \citep{van2000asymptotic}. Our proof of Theorem~\ref{theorem1} utilizes Lemma 2 \citep{meisner2021combining}  to first show that the empirical smooth approximations of TPR and FPR  converge almost surely to the smooth versions of TPR and FPR respectively over a fixed parameter space. We then show that the smooth versions of TPR and FPR converge to TPR and FPR as $h\rightarrow 0$. We finally show c) by employing results from M-estimation theory (theorem 5.7 of \cite{van2000asymptotic}) in addition to the results in b). \albert{In showing consistency of $\hat{\bm \beta}^{\Phi}$ in c), we assume $\bm \beta^*$ is a unique and well-separated maximizer of $TPR(\bm \beta)$ over the parameter space}. Details of the proof can be found in the Supplementary Appendix D.

 The general implication of the above theorem is that, even in cases where the global optimal rule is not linear, we can guarantee performance of our proposed linear rule to be asymptotically close to the best linear rule within the class of linear rules.  As mentioned earlier, this is not guaranteed using convex surrogate functions.
 
\section{Numerical studies} \label{simulation}
In this section, we compare the finite sample performance of the proposed methods, namely the plug-in approaches, DOOLR, and IT-DOOLR against the standard logistic regression (`Standard') approach. The Standard approach is implemented in two steps; i) fit a logistic regression model for the outcome; ii) choose a threshold for the predicted probabilities that corresponds to the PPV constraint.
For the plug-in approach, we considered two ways of estimating the nuisance risk function, $\text{pr}(D=1|\bm X)$: 1) logistic regression (`Plug-in Logistic'), and 2) ensemble of methods (generalized additive models, logistic regression, random forest,  mean of outcome) via \texttt{SuperLearner} (`Plug-in SL'). \albert{For our proposed linear rules (DOOLR and IT-DOOLR), we set, across all simulations, $h = n^{-1/3}\,\mathrm{StdDev}\!(\mathcal{X}^{T}\hat{\bm \beta}^{0}/\lVert \hat{\bm \beta}^{0}\rVert),$
where $\hat{\bm \beta}^{0}$ denotes the coefficient estimate obtained from the Standard approach. We present a numerical study examining different values of $h$ to demonstrate the robustness of our decision rule in Supplementary Appendix~C (see Table~2). To mimic a rare disease screening scenario, we generate data with disease prevalence approximately $1\%$ across all simulation settings and repeat each simulation 500 times. We specify PPV constraints at $3\%$, $4\%$, and $4.5\%$ in all simulations. We present only the results for $4\%$ here and report the remaining results in Supplementary Appendix~C (see Tables~4-7).} Our implementation of the proposed methods and codes to replicate simulation results in R is available at \url{https://github.com/albertosom/BiomarkerRule}.

\begin{table}[!t]
\caption{Cohort study sampling simulation scenario when the true decision rule is linear, piece-wise linear, and nonlinear. Average TPR and PPV with corresponding standard error (in parentheses) in the test sample across 500 simulation replicates under PPV $(\alpha=0.04)$ and prevalence $(p_1=0.01)$.}
\label{tab:1}
\centering
\scalebox{0.8}
{\begin{tabular}{ccccccccc}
\toprule
\multicolumn{1}{c}{$n$} & \multicolumn{1}{c}{ } & \multicolumn{1}{c}{\textbf{Measure}}& \multicolumn{4}{c}{\textbf{Methods}} \\
\cmidrule(rl){4-7} 
\multicolumn{3}{c}{} &{Standard} & {Plug-in Logistic} & {Plug-in SL}  & {DOOLR} \\
\midrule
 & &  \textbf{Linear} &  &  &  & \\\\
2500 & & TPR &  0.988(0.006) & 0.987(0.010)  & 0.987(0.013)  & 0.967(0.040)  \\
   & & PPV & 0.042(0.003)  & 0.040(0.003)  &  0.038(0.012) & 0.047(0.015)  \\\\
5000 & & TPR & 0.990(0.004) & 0.990(0.006) &  0.990(0.006) & 0.980(0.017)
 \\
   & & PPV & 0.042(0.001)  & 0.040(0.001)  &  0.040(0.005) & 0.043(0.008)  \\\\
 & &  \textbf{Linear with contamination} &  &  &  & \\
2500 & & TPR &  0.508(0.284) & 0.547(0.281)  & 0.972(0.020)  & 0.922(0.123)  \\\\
   & & PPV & 0.043(0.025)  & 0.032(0.014)  &  0.040(0.003) & 0.045(0.008)  \\\\
5000 & & TPR & 0.571(0.284) & 0.651(0.295) &  0.977(0.011) & 0.959(0.042) \\
  &  & PPV &  0.05(0.022) & 0.037(0.008) & 0.040(0.002) & 0.043(0.003)\\\\
 & &  \textbf{Piece-wise linear} &  &  &  & \\\\
2500 & & TPR &  0.604(0.152) & 0.678(0.116)  & 0.771(0.159)  & 0.624(0.163)  \\
   & & PPV & 0.055(0.021)  & 0.043(0.011)  &  0.045(0.019) & 0.050(0.019)  \\\\
5000 & & TPR & 0.628(0.144) & 0.693(0.085) &  0.840(0.079) & 0.654(0.115) \\
  &  & PPV &  0.051(0.020) & 0.042(0.007) & 0.041(0.006) & 0.046(0.013)\\\\
  & &  \textbf{Nonlinear} &  &  &  & \\\\
2500 & & TPR &  0.696(0.256) & 0.833(0.201)  & 0.984(0.033)  & 0.747(0.218)  \\
   & & PPV & 0.047(0.010)  & 0.041(0.006)  &  0.040(0.004) & 0.042(0.009)  \\\\
5000 & & TPR & 0.601(0.296) & 0.868(0.169) &  0.989(0.007) & 0.808(0.189) \\
  &  & PPV &  0.048(0.009) & 0.041(0.004) & 0.040(0.001) & 0.043(0.006)\\
\bottomrule
\end{tabular}}
\end{table}

\subsection{Linear decision rule} \label{cohort-linear} 
We consider bivariate normal biomarkers with and without contamination, similar scenarios are described in \cite{ croux2003implementing} and \cite{meisner2021combining}. We first simulate the covariates $\textbf{X}=(X_1,X_2)$, each following a standard normal distribution. The outcome $D$ is generated from a Bernoulli distribution with success probability  $\text{pr}(D=1|X_1,X_2)=1/(1+ \text{exp}(\beta_0 +X_1\beta_1 + X_2\beta_2))$. The optimal decision rule has a linear form and is expressed as $d(\bm X)=\mathbbm{1} \{\beta_{0} + \beta_{1}X + \beta_{2}X_{2} + c >0 \} $, where $c$ is chosen to satisfy pre-specified PPV-contraint.  We set the parameters $\beta_0=-8.7$, $\beta_1=2.4$, and $\beta_2=2.4$ so that the disease prevalence is approximately $1\%$. We further added 6\% contaminated control ($D=0, X_1=6$  and $X_2=6$) observations to depict the scenario where patients with  biomarker values similar to those of cases are labeled controls. We estimate the decision rule using training samples of sizes $n_{train}= 2500$ and $5000$. The results were then evaluated on an independent test sample of size $n_{test}= 10^{5}$. 
 
 The simulation results are presented in Table \ref{tab:1}. When there are no contaminated observations, we observe that the Standard approach has the best performance. This is expected as the data is generated directly from a logistic model. The performance of the DOOLR approach improves as sample size increases. The two Plug-in approaches demonstrated comparable performance to the Standard approach. In the  presence of contaminated observations, the proposed methods provided better control of the PPV compared to the Standard approach. For the estimate of TPR, the Plug-in SL approach achieved the best performance, likely due to its flexibility in modeling the nuisance risk function. Also, the developed linear rule via the DOOLR method performed remarkably well, yielding results comparable to the Plug-in SL, \albert{while offering the added advantage of a more interpretable decision rule due to its linear structure}. This further illustrates its robustness in the presence of outlier observations compared to alternative methods.  
 
\subsection{Nonlinear decision rule} \label{nonlinear}
In this simulation setting, we consider two scenarios where the true clinical decision rule is not linear. The first scenario we consider is when the true rule is piecewise-linear. Suppose we have two independent biomarkers $(X_1, X_2)$, each generated from a standard normal distribution. The outcome $D$, is generated as $D= \mathbbm{1}\{ [(-8.9 +2X_1 + 2 X_2 *\mathbbm{1}\{X_2> X_2^{0.025}\}  + \varepsilon] >0 \}$ where 
$\varepsilon$ is distributed as a logistic random variable with location parameter zero and scale parameter one, and $X_2^{0.025}$ is the $0.025$th quantile of $X_2$. Among observations with $\mathbbm{1}\{X_2< X_2^{0.025}\}$, we randomly sample $0.4\%$ and set $D=1$. This setting  mimics a practical scenario where the true decision rule is piece-wise linear. For example, in PDAC screening, it has been shown that higher values of serum biomarker CA19-9 in addition to other patient characteristics (e.g., A1c, age, etc.) are generally associated with higher risk of PDAC \citep{ballehaninna2012clinical}. However, a small group of patients (e.g., $\text{Lewis}^{\alpha-\beta-}$ genotype) with very low values of CA19-9 are also at high risk of PDAC \citep{yong2022low}. 

For the second scenario, we consider when the true clinical decision rule is a nonlinear function of the covariates. Suppose we have three independent biomarkers $\bm{X}=(X_1, X_2,X_3)$, each generated from the standard normal distribution. The outcome $D$, is generated as $D= \mathbbm{1}\{ [(\beta_{0} +\beta_{1}\text{sin}(X_1) + \beta_{2}X_{2}^{2} + \beta_{3}\text{cos}(X_{3})+ \varepsilon] >0 \}$ where $\varepsilon$ is distributed as a logistic random variable with location parameter zero and scale parameter one. We set the parameters $\beta_0=-8.6$, $\beta_1=5$, $\beta_2=-4$, and $\beta_{3}=3$ so that the disease prevalence is approximately $1\%$. We train the decision rules on training samples of sizes $n_{train}= 2500, 5000$ and evaluate the performance on an independent test sample $n_{test}= 10^5$. We train our decision rules on training samples of sizes $n_{train}= 2500, 5000$ and evaluate the performance on an independent test sample $n_{test}= 10^5$. 

The results are presented in Table \ref{tab:1}. All proposed methods yield higher TPR estimates  compared to the standard logistic regression approach. Notably, the superior performance of DOOLR highlights its optimality among the class of linear rules. In terms of the PPV control, the proposed methods also demonstrate favorable results. Among all methods, the Plug-in SL approach has the best performance, due to its flexible estimation of the nuisance risk function. 

\subsection{Incorporating existing risk information}
In this simulation setting, we investigate the performance of IT-DOOLR under three scenarios where existing risk information is available. We generate data using the same data generation mechanisms under the nonlinear rule described in Section \ref{nonlinear}, and compare the performance of the IT-DOOLR approach, which incorporates the existing information, to the DOOLR method. We let the existing model be parameterized by $\widetilde{\beta}$ such that the decision rule based on the existing model is $\mathbbm{1}\{ (1,H(\bm X))^{T}\bm \widetilde{\beta} >0\}$, where $H(\bm X)=(\text{sin}(X_1), X_{2}^{2}, \text{cos}(X_{3}))$. We train the decision rules on training datasets with sample sizes of 2500 and 5000, and evaluate the performance on a testing dataset of size of $10^{6}$.  We consider the following three scenarios
\begin{itemize}
    \item Scenario I: the existing risk model is the same as the true model under the data generation $\bm{\widetilde{\beta}} = \bm{\beta}=(-8.6,5,-4,3)$
    \item Scenario II: the existing risk model does not include all covariates under the true data generation mechanism, here $\bm{\beta}=(-8.6,5,-4,3)$ and $\bm{\widetilde{\beta}}=(-8.6,5,0,0)$ with two covariates omitted in the existing model.
    \item Scenario III: the existing risk model is very different from the model under the true data generation mechanism, here $\bm{\beta}=(-8.6,5,-4,3)$ and $\bm{\widetilde{\beta}}=(-15,-3,-1, -1)$.
\end{itemize}

We generate outcome for all three scenarios using the same data generation mechanism described in Section~\ref{nonlinear}. The results are summarized in Table~\ref{tab:2}. In Scenario I, where the existing model aligns with the true data-generating process, IT-DOOLR effectively leverages the external risk information to improve upon the performance of DOOLR. In Scenario III, where the existing risk model deviates from the true data-generating mechanism, IT-DOOLR still achieves a modest improvement in TPR estimation. Overall, the proposed integration approach effectively utilize valuable external information (Scenarios I \& II) while avoiding negative information transfer (Scenario III) when the existing model is potentially misspecified. In particular, IT-DOOLR exhibits greater stability and improved estimation of both TPR and PPV compared to DOOLR, especially in scenarios where the external information is informative.

\begin{table}[!t]
\caption{Simulation scenarios under incorporating auxiliary information. Average TPR and PPV with corresponding standard error (in parentheses) in the test sample across 500 simulation replicates under PPV $(\alpha=0.04)$ and prevalence $(p_1=0.01)$.}
\label{tab:2}
\centering
\scalebox{0.8}
{\begin{tabular}{ccccccc}
\toprule
\multicolumn{1}{c}{$n$} & \multicolumn{1}{c}{\textbf{Measure}}& \multicolumn{1}{c}{\textbf{DOOLR}}& \multicolumn{3}{c}{\textbf{IT-DOOLR}} \\
\cmidrule(rl){4-6} 
\multicolumn{3}{c}{}  & {Scenario I} & {Scenario II}  & {Scenario III}  \\
\midrule
2500 & TPR & 0.747(0.218)  & 0.839(0.136)  & 0.802(0.165)  & 0.781(0.179)  \\\\
   & PPV &  0.043(0.009) & 0.041(0.006)  &  0.041(0.007)& 0.043(0.007)  \\\\
5000 & TPR &  0.808(0.189) & 0.880(0.097)  & 0.860(0.118) & 0.845(0.145)   \\\\
   & PPV & 0.043(0.006)  &  0.041(0.004)  & 0.042(0.005)  & 0.042(0.006) \\\\
\bottomrule
\end{tabular}}
\end{table}
We additionally consider a scenario where the data is obtained from a nested case-control sampling design, commonly considered in biomarker studies \citep{ernster1994nested}. The performance among the methods is similar to the performance under the simulation scenario in Section \ref{cohort-linear}. The data generation mechanism and results are summarized in Supplementary Appendix~C Table~1.

\section{Application to PANDORA study data}\label{application}
Early detection of PDAC can greatly improve patients' survival. Unfortunately, PDAC  lacks an effective strategy for early detection due to several challenges. Notably, the low incidence of PDAC in the general population (10/100,000) makes it cost-prohibitive to implement a biomarker-based screening program for asymptomatic individuals. Even if a biomarker test has high sensitivity and specificity, it may still cause thousands of false positives for every identified positive cancer. Therefore, biomarker-based screening should be applied in a population that has a higher risk for PDAC. Particularly, new Onset Hyperglycemia and Diabetes (NOD) patients have shown a higher risk for PDAC \citep{chari2005probability}. 

Leveraging data from Kelsey-Seybold Clinic (Houston, Texas) during the period 2011- 2022, the  PANDORA  Cohort includes a total of 357,593 patients with data on either their glycated haemoglobin (HbA1c) or fasting blood glucose (FBG) level. Out of this number, we obtained the NOD population by selecting patients based on the first occurrence of any of the following; a) HbA1c $\ge 6.5\%$, b) two consecutive FBG $\ge 126 mg/dl$ within 18 months of each other, 
or c)	FBG $\ge 126 mg/dl$ followed by the start of diabetes medication within 30 days. Our total NOD population consists of 18,300 patients, of whom 35 developed PDAC within three years. Among the 35 confirmed PDAC cases, 32 have complete information on weight, and estimated glucose levels at both NOD date and left window (approximately 1 year prior to the NOD date). Among the 18265 controls, 11,149 have complete three years follow-up of which 7,137 have complete  biomarker information at both NOD date and left window.

The goal of our analysis is to develop a new biomarker combination rule that incorporates broader covariate information to identify NOD patients at high risk of developing PDAC and recommend them for further screening. Particularly, the rule should be developed under a practical constraint on the PPV.  \cite{sharma2018model} proposed an algorithm (the ENDPAC model) to stratify NOD patients into high or low PDAC risk based on changes in FBG levels, weight, and age at NOD onset. While the model demonstrated promising performance in their validation cohort, efforts to validate it in other populations have shown lower performance \citep{chen2021validation,boursi2022validation}. Several factors may contribute to this reduced performance: (1) the original ENDPAC model was developed primarily in a predominantly white population, and (2) the model relies on FBG measurements, which are often unavailable for many patients. To address the latter limitation, the authors proposed a heuristic conversion between HbA1c and FBG, given by $(\text{FBG}= 28.7\times \text{HbA1c}-46.7)$, for patients without FBG values. However, this conversion is suboptimal, as it fails to fully capture the nuances of FBG and HbA1c measurements. Specifically, FBG reflects glucose control at the time of the blood test and exhibits higher day-to-day variability, whereas HbA1c represents an average blood glucose level over the past 2–3 months and thus has lower day-to-day variability. We incorporate risk information from the ENDPAC model to assess whether it improves the performance of our biomarker combination rule. 

Our analysis data consist of the 32 diagnosed PDAC cases and  640 random sample of controls from a total NOD cohort of size 18,300. For each case, we randomly sample 20 controls from individuals in the cohort who have not developed the outcome by the time the case occurred, this depicts a nested case-control sampling design scenario. We construct the rule based on patients' change in HbA1c, change in weight and age at NOD. We standardized these variables before model fitting. To assess the performance of our method, we randomly split the data equally into training and testing sets. We develop the proposed methods on the training data and evaluated their performance on the independent test set. This random split was repeated 200 times, and the average of the estimates on the test sets is reported. For our analysis, we use a prevalence estimate of PDAC in the NOD population of $p_1=0.44\%$, and pre-specified PPV values of $1\%$, $1.5\%$, and $2\%$, corresponding to screening approximately 100, 67, and 50 individuals to detect one cancer, respectively. From the study data, we  calculate ENDPAC score for each patient and identify risk thresholds to define high- and low-risk groups that satisfy the PPV constraints. 
 
\begin{table}[!t]
\caption{Application to screening NOD patients at high risk of PDAC data given $2\%$ constraint on PPV using an estimate of prevalence of pancreatic cancer to be $0.44\%$. We report average TPR, PPV, and estimate of coefficients with corresponding standard error (in parentheses). The linear decision rule from IT-DOOLR is $\mathbbm{1}\{0.453 + 0.208\times  \mathrm{Age}  - 0.677 \times \mathrm{Weight change} + 0.672 \times \mathrm{Change in HbA1c} >0\}$.}         
\label{tab:3}
\centering
\scalebox{0.8}
{\begin{tabular}{ccccccc}
\toprule
\multicolumn{1}{c}{\textbf{ }} & \multicolumn{1}{c}{}& \multicolumn{4}{c}{\textbf{Methods}} \\
\cmidrule(rl){3-7} 
\multicolumn{2}{c}{} & {Standard} & {Plug-in Logistic} & {Plug-in SL}  & {DOOLR} & {IT-DOOLR}  \\
\midrule
 
     &  &  &    $\mathbf{\alpha=2\%}$&   &  & \\
\textbf{Measure} & 
    TPR & 0.682(0.114) & 0.685(0.125) & 0.760(0.101) & 0.755(0.124) & 0.771(0.122)  \\\\
   & PPV & 0.023(0.010) & 0.024(0.010) & 0.020(0.010) & 0.020(0.007) & 0.020(0.006) \\\\
\textbf{Predictors} &  &  &  &   &  & \\
   & Age &  0.444(0.113) &  &  & 0.204(0.098) & 0.208(0.097)\\
& Weight change &  -0.805(0.086) &  & & -0.6828(0.127)& -0.677(0.122)\\
    & Change in HbA1c &  0.355(0.095) &  &  & 0.678(0.147) &  0.672(0.149) \\

\bottomrule
\end{tabular}}
\end{table}
We present the results of our analysis based on the proposed methods in Table \ref{tab:3}. We considered the PPV constraint at 2\%, to illustrate the proposed method.  The reported values are the average estimates (with standard deviations) of TPR and PPV across the test sets. To achieve a PPV constraint of $2\%$,  the threshold for the ENDPAC score is set at 2, and the corresponding TPR of the ENDPAC score based rule is 68.8\%. The Plug-in Logistic perform similarly to the Standard approach. The IT-DOOLR estimate of TPR is $77.1\%$, slightly higher than the TPR estimate for DOOLR $(75.5\%)$, indicating some information transfer from the ENDPAC model. Both DOOLR and IT-DOOLR yield PPV estimates that closely align with the pre-specified value of $2\%$. 

Finally, Table \ref{tab:3} also presents  the estimated coefficients of the predictors.  These coefficients can be used to construct linear decision rules (DOOLR and IT-DOOLR) by comparing linear combinations of covariates and coefficients to thresholds of -0.472 and -0.453, respectively. A reduction in weight between the NOD date and the left window is associated with a higher risk of PDAC. This aligns with the current research findings, which suggest that among individuals with new-onset diabetes, unexplained weight loss may  indicate underlying pancreatic cancer rather than just the metabolic changes from diabetes itself. Result for the other PPV values are reported in Supplementary Appendix~C Table~2.

\section{Discussion}\label{discussion}

This paper proposes a new PPV-constrained benefit function for developing individualized clinical decision rules, designed to maximize TPR while ensuring adherence to a clinically meaningful pre-specified PPV. We further extend this framework to integrate relevant external risk information from published literature. Our proposed linear and nonlinear rules offer practitioners the flexibility to choose between interpretability and predictive power. When interpretability is prioritized, we show that DOOLR is guaranteed to estimate the best rule within the class of linear rules. On the other hand, the plug-in approach allows a more flexible estimation of the nonlinear rule, serving as an attractive alternative when interpretability is less of a concern and predictive accuracy is more important. The strength of our framework is that both of our approaches are model-free which can be robust when the risk model is mis-specified.

A key distinction between our proposed objective and prior methods that maximize TPR under an FPR constraint \citep{wang2020learning,meisner2021combining} lies in the complexity of PPV. Constraining FPR at a pre-specified value to achieve a target PPV is generally intractable because PPV is a nonlinear function of TPR, FPR, and disease prevalence. \albert{Instead, our approach directly maximizes TPR while incorporating FPR and prevalence from the same dataset, eliminating the need for an explicit FPR constraint and aligning the optimization with the intended clinical utility: enforcing a PPV constraint that reflects downstream diagnostic burden while optimizing sensitivity. This results in a fundamentally different optimization problem from FPR-constrained approaches and provides a more relevant framework for evaluating and deploying biomarker-based screening tools.}

Additionally, the developed biomarker decision rules can be broadly applicable to various early detection initiatives. For instance, in ovarian cancer screening, clinical guidelines recommend that a biomarker-based screening rule achieve a minimum PPV of $10\%$ to justify invasive surgical diagnostic procedures \citep{jacobs2004progress}. Using our proposed method, an optimal biomarker-based decision rule for ovarian cancer can be developed while ensuring the PPV is maintained at the desired $10\%$ level. Similarly, \cite{mazzone2024clinical} investigated biomarker combination rule for the early detection of lung cancer, proposing a decision rule based on a penalized logistic regression model, which yielded a PPV of  $1.3\%$. They reported their rule reduced the number needed to screen with low-dose computed tomography (CT) to detect one lung cancer from 143 to 75. However, while their rule is straightforward, it may not be optimal for maximizing the screening benefit (i.e., TPR). Our proposed nonparametric approach provides a rigorous alternative to developing optimal decision rules in such settings, ensuring both improved screening efficiency and adherence to clinically relevant constraints. More broadly, this method offers a valuable tool for medical decision-making, accommodating patient heterogeneity while addressing practical implementation challenges in clinical practice.

\albert{In practice, PPV threshold should be  specified a priori based on clinical and operational considerations. Particularly, PPV directly determines the number needed to screen to detect one cancer case, which quantifies the
downstream diagnostic burden. A clinically meaningful PPV threshold should
reflect the maximum number of false-positive diagnostic evaluations that the
health system can support, given the cost, availability, and invasiveness of
confirmatory imaging or biopsy. For PDAC, where prevalence is extremely low (less than $1\%)$, PPV levels in the range of 1–5\% correspond to realistic   NNS values for CT/MRI follow-up, whereas higher PPV levels may be needed to justify invasive procedures.}

With our choice of non-convex surrogate function, one would have to take some care in choosing initial values for the optimization. Although we found using the  coefficient estimates from standard logistic regression as initial value to work generally well, one could also implement the optimization algorithm with different multiple initial values to ensure stability. Other choice of surrogate function that could be considered is the ramp loss \citep{huang2014identifying} which can be implemented using the difference of convex functions algorithm (DCA) \citep{tao1996numerical}. This could potentially provide a more reliable optimization approach as it decomposes the problem into convex sub-problems, which are generally easier to solve than non-convex ones. 

\albert{An interesting direction for future work is to incorporate ranking-based penalties when integrating external risk information. While our proposed method penalizes disagreement in binary decisions, one could instead encourage agreement in pairwise risk ranking in settings where only relative risk ordering is available. This may provide a principled alternative for incorporating external information.}

\section*{Acknowledgments}
The work was supported by a Catalyst Award from the Pancreatic Cancer Action Network. We express our appreciation to Lynn Matrisian for helpful comments.

\bibliographystyle{apalike}
\bibliography{reference}

 \clearpage

\begin{center}
  \Large Supplementary Material to ``\emph{Developing Performance-Guaranteed Biomarker Combination Rules with Integrated External Information under Practical Constraint}"
\end{center}

\section{Appendix A (Optimality of decision rule)} \label{optimality}

Let the objective and constraint functions be denoted as
\begin{align*}
f(d(\bm X)) 
&= \mathbb{E}_{\bm X}\!\left[\mathbbm{1}\{d(\bm X)=1\}\,\eta_1(\bm X)\right], \text {and } \\
h(d(\bm X)) 
&= \mathbb{E}_{\bm X}\!\left[\mathbbm{1}\{d(\bm X)=1\}
\left\{\gamma\eta_1(\bm X)-\alpha\gamma\eta_1(\bm X)-\alpha\eta_0(\bm X)\right\}\right], \text{ respectively}.
\end{align*}

For $\lambda \ge 0$, we define the Lagrangian
\begin{align*}
L(d(\bm X),\lambda)
= E\!\left[\mathbbm{1}\{d(\bm X)=1\}
\left\{\eta_1(\bm X)
-\lambda\alpha\gamma\eta_1(\bm X)
-\lambda\alpha\eta_0(\bm X)
+\lambda\gamma\eta_1(\bm X)\right\}\right],
\end{align*}
and its associated dual function
\begin{align*}
g(\lambda) = \sup_{d(\bm X)} L(d(\bm X),\lambda).
\end{align*}

For any fixed $\lambda$, the Lagrangian is maximized pointwise, yielding the decision rule
\begin{align*}
d_\lambda(\bm X) = \mathbbm{1}\{f_\lambda(\bm X) > 0\},
\end{align*}
where
\begin{align*}
f_\lambda(\bm X)
= \eta_1(\bm X)
-\lambda\alpha\gamma\eta_1(\bm X)
-\lambda\alpha\eta_0(\bm X)
+\lambda\gamma\eta_1(\bm X).
\end{align*}

Our goal is to establish strong duality, namely
\begin{align*}
p^*
= \sup_{d \in \mathscr{R}:\, h(d)\ge 0} f(d)
= \inf_{\lambda \ge 0} g(\lambda)
= g^*.
\end{align*}

We note that, weak duality holds immediately: for any decision rule $d$ satisfying $h(d)\ge 0$ and any $\lambda \ge 0$,
\begin{align*}
f(d) \le L(d,\lambda) \le g(\lambda),
\end{align*}
which implies $p^* \le g^*$.

To establish equality, we make the following assumptions:
\begin{enumerate}
\item There exists a strictly feasible decision rule $\tilde d$ such that $h(\tilde d) > 0$.
\item $\Pr\!\left(f_\lambda(\bm X)=0\right)=0$ for all $\lambda \ge 0$, ensuring that $d_\lambda$ is almost surely well-defined.
\end{enumerate}

When $\lambda=0$ (unconstrained problem), $f_{0}(\bm X)= \eta_{1}(\bm X)$. Since $\eta_{1}(\bm X)>0$, the PPV associated with the decision rule, $d_{0}(\bm X)$, is equal to the prevalence which we assume is not greater than the pre-specified, $\alpha$. Denote the decision rule with $\lambda=0$ by $d_0$, we have
\begin{align*}
h(d_0) < 0.
\end{align*}
Since the integrand defining $h(d_\lambda)$, $\mathbbm{1}\{d_\lambda (\bm X)=1\}
\left\{\gamma\eta_1(\bm X)-\alpha\gamma\eta_1(\bm X)-\alpha\eta_0(\bm X)\right\}$, is bounded and measurable, $h(d_\lambda)$ is continuous in $\lambda$ by the dominated convergence theorem.  Therefore, by the intermediate value theorem, there exists $\lambda^* > 0$ such that
\begin{align*}
h(d_{\lambda^*}) = 0.
\end{align*}

At this value, $d_{\lambda^*}$ is feasible for the primal problem, so
\begin{align*}
p^* \ge f(d_{\lambda^*}).
\end{align*}
Since $d_{\lambda^*}$ maximizes the Lagrangian at $\lambda^*$,
\begin{align*}
f(d_{\lambda^*}) = L(d_{\lambda^*},\lambda^*) = g(\lambda^*).
\end{align*}
Finally, because $g(\lambda^*) \ge g^* \ge p^*$, all inequalities hold with equality, implying
\begin{align*}
p^* = g^*.
\end{align*}

This establishes strong duality and justifies the use of the Lagrangian formulation in our setting.

\section{Appendix B (Optimization algorithm and justification)}

Consider the population version of the constraint optimization problem; 
\begin{align}
    \text{maximize }  &\E[D\Phi\left(\frac{\mathscr{X}_i^T\bm{\beta}}{h}\right)/p_{1}]  \nonumber \\
    \text{ subject to } & \E[D\Phi\left(\frac{\mathscr{X}_i^T\bm{\beta}}{h}\right)/p_{1}] -\alpha \E[\Phi\left(\frac{\mathscr{X}_i^T\bm{\beta}}{h}\right)] \ge 0 \label{optim1}
\end{align}

The Lagrangian corresponding to $\eqref{optim1}$ is 
\begin{align}
    L^{\Phi}(\bm \beta,\lambda) = \E[D\Phi\left(\frac{\mathscr{X}_i^T\bm{\beta}}{h}\right)/p_{1}]  - \lambda(\alpha \E[\Phi\left(\frac{\mathscr{X}_i^T\bm{\beta}}{h}\right)] - \E[D\Phi\left(\frac{\mathscr{X}_i^T\bm{\beta}}{h}\right)/p_{1}]  ) \label{lag2}
\end{align}
for Lagrange multiplier $\lambda \in [0,\infty]$. Denote $\kappa=\frac{\lambda}{1 +\lambda} \in [0,1]$, then for some fixed $\kappa$, our goal is solve for $\beta^{\Phi}_{\kappa}$ where 
\begin{align}
  \bm{\beta}_{\kappa}^{\Phi}= \underset{\beta}{\text{arg max }} \E[ \{D(1-\kappa)  - \kappa(\alpha p_{1}  - D)\}\Phi\left(\frac{\mathscr{X}_i^T\bm{\beta}}{h}\right)  ]  \label{optim2}
\end{align}
Let $\mathcal{P}( \kappa)$ and $\mathcal{B}( \kappa)$ denote the expected benefit and PPV constraint associated with the optimal decision rules from $\eqref{lag2}$
\begin{align*}
    \mathcal{B}(\kappa ) & = \E[D\Phi\left(\frac{\mathscr{X}_i^T\bm{\beta}_{\kappa}^{\Phi}}{h}\right)/p_{1}] \\
    \mathcal{P}(\kappa) & = -\alpha \E[\Phi\left(\frac{\mathscr{X}_i^T\bm{\beta}_{\kappa}^{\Phi}}{h}\right)] -\E[D\Phi\left(\frac{\mathscr{X}_i^T\bm{\beta}_{\kappa}^{\Phi}}{h}\right)/p_{1}]
\end{align*}

If it holds that  $\mathcal{P}(1) <0 \le \mathcal{P}(0)$, we can obtain some $\kappa^{*}$ ( may not be unique) such that  $\mathcal{P}(\kappa^{*})=0$. The restriction  $\mathcal{P}(\kappa) <0$ ensures that there is a feasible solution to the constraint and $ \mathcal{P}(\kappa) \ge 0$ is to restrict the case when the PPV under the unconstrained problem is larger than $\alpha$.

The above  indicates that to solve $\eqref{optim1}$ , we only need to solve the unconstrained problem in \eqref{optim2} for any $\kappa$ and then find $\kappa^{*}$ such that the PPV associated with $\kappa^{*}$ satisfies the constraint.

\begin{algorithm*}[t] 
\caption{Estimating optimal linear decision rule}
\begin{algorithmic}[1]
  \State Given training data $(D_{i}, \bm X_{1i})$ for $i=1,\ldots,n$ and pre-specified $\alpha$.
   \State Start with an initial value of $\bm \beta$,  i.e. $ \bm \hat{\bm \beta}_{0}$, which one can derive from the standard logistic regression approach. 
    \State Obtain  
    \begin{align*}
  \hat{\bm{\beta}}_{\kappa}^{\Phi}= \underset{\beta}{\text{arg max }} \sum_{i=1}^{n}[ \{D(1-\kappa)  - \kappa(\alpha p_{1}  - D) \}\Phi\left(\frac{\mathscr{X}_i^T\bm{\beta}}{h}\right)  ]  \label{optim2}
\end{align*} using any non-linear optimization (e.g., \texttt{nlm} in R) of your choice.
    \State Perform grid search on $[0,1]$ for $\hat{\kappa}$ that solves 
    \begin{align*}
        \widehat{\text{PPV}}(\mathbbm{1}\{\mathscr{X}^T\hat{\bm{\beta}}_{\kappa}^{\Phi}>0\} )-\alpha=0.
    \end{align*}
    
    \State Project   $\hat{\bm{\beta}}_{\hat{\kappa}}^{\Phi}$ to satisfy the  $\|\hat{\bm{\beta}}_{\hat{\kappa}}^{\Phi} \|^{2}_{2}=1$.
    \State Final linear rule is of the form $ \hat d_L(\bm X,\hat{\bm \beta}^{\Phi}_{\kappa})= \mathbbm{1}\{\mathscr{X}^T\hat{\bm{\beta}}_{\kappa}^{\Phi}>0\}$.
\end{algorithmic}
\label{Algo1}
\end{algorithm*}

\section{Appendix C (Additional numerical study and data application)}

\subsection{Nested case-control sampling} \label{case-control}

Here we consider the scenario where data is obtained from a nested case-control sampling design, commonly considered in biomarker studies \citep{ernster1994nested}. We first simulate a large cohort  of size $N=10^6$, with  covariates $\textbf{X}=(X_1,X_2)$ simulated from the standard normal distribution, and the outcomes $D$ are generated from the conditional distribution $\text{pr}(D=1|X_1,X_2)=1/(1+ \text{exp}(\beta_0 +X_1\beta_1 + X_2\beta_2))$. A similar simulation setting is illustrated in \cite{rose2008note}. We set $\beta_0=-8$, $\beta_1=2.1$, and $\beta_2=2.1$ such that the disease prevalence is approximately $1\%$. We further added 6\% contaminated control ($D=0, X_1=6$  and $X_2=6$) observations to depict the  setting where the data contains some controls that are actually cases.  We then obtain the case-control data by randomly sampling $n_{1}$ cases and $n_{0}=20n_{1}$ controls resulting in a sample size of $n=n_{1}+n_{0}$. We implement our methods on training samples of size $n_{train}= 2100$ and $4200$. Performance was evaluated on an independent test sample of size $n_{test}= 21,000$. The results are presented in Table~\ref{tab:casecontrol}.

\begin{table}[!t]
\caption{Nested case-control study sampling scenario with $6\%$ contaminated control observations. Average TPR and PPV with corresponding standard error (in parentheses) in the test sample across 500 simulation replicates under PPV $(\alpha=0.04)$ and prevalence $(p_1=0.01)$.}
\label{tab:casecontrol}
\centering
\scalebox{0.8}
{\begin{tabular}{cccccc}
\toprule
\multicolumn{1}{c}{$n$} & \multicolumn{1}{c}{\textbf{Measure}}& \multicolumn{4}{c}{\textbf{Methods}} \\
\cmidrule(rl){3-6} 
\multicolumn{2}{c}{} &{Standard} & {Plug-in Logistic} & {Plug-in SL}  & {DOOLR}  \\
\midrule
2100 & TPR & 0.710(0.184) & 0.591(0.273) &  0.956(0.009) & 0.932(0.045) \\
  & PPV &  0.053(0.010) & 0.037(0.007) & 0.040(0.003) & 0.042(0.004)\\\\
4200 & TPR & 0.770(0.194) & 0.654(0.301) & 0.957(0.007) & 0.945(0.019)  \\
   & PPV & 0.053(0.008) & 0.040(0.006) & 0.039(0.003) & 0.042(0.003) \\
\bottomrule
\end{tabular}}
\end{table}

\subsection{Results under different PPV constraint}
We present simulation and application results for scenarios in the manuscript for different PPV constraints.
\begin{table}[h]
\caption{Application to identifying NOD patients at high risk of PDAC. Estimate of prevalence of pancreatic cancer in the PANDORA study is $0.44\%$. We pre-specify PPV constraint values at $1\%$, $1.5\%$, and $2\%$ and report average TPR, PPV with corresponding standard error (in parentheses). When the PPV constraint is low at $1\%$, the Standard approach performed better than our proposed linear rules. However, as the PPV constraints increase, we observe reduction in performance of the Standard approach.}
\centering
\label{tab:application}
\scalebox{0.85}
{\begin{tabular}{ccccccccc}
\toprule
$\alpha$ &  & Measure & Standard & Plug-in Logistic & Plug-in SL & DOOLR & DOOLR-IT \\
\midrule
0.010 & & TPR & 0.964(0.048) & 0.969(0.043) & 0.961(0.055) & 0.906(0.068) & 0.912(0.069) \\\\
      & & PPV & 0.010(0.001) & 0.010(0.001) & 0.010(0.001) & 0.010(0.002) & 0.010(0.001) \\\\
\midrule
0.015 & & TPR & 0.807(0.110) & 0.817(0.097) & 0.848(0.092) & 0.830(0.111) & 0.839(0.099) \\\\
      & & PPV & 0.016(0.006) & 0.015(0.004) & 0.015(0.003) & 0.016(0.004) & 0.015(0.003) \\\\
\midrule
0.020 & & TPR & 0.659(0.124) & 0.671(0.127) & 0.745(0.110) & 0.755(0.106) & 0.775(0.105) \\\\
      & & PPV & 0.024(0.012) & 0.024(0.015) & 0.020(0.006) & 0.020(0.006) & 0.020(0.004) \\\\
\bottomrule
\end{tabular}}
\end{table}
\begin{table}[h]
\caption{Data generation is under the linear rule with contamination scenario ( $\beta_{0}=-12.5$, $\beta_1=\beta_2=7, p_{1}=0.1, \text{and } \alpha=0.3$): mean and standard deviation (in parentheses) of TPR and PPV for the DOOLR rule across smoothing parameters $h$. Here, Adaptive refers to 
$h = n^{-1/3}\,\mathrm{StdDev}\!(\mathcal{X}^{T}\hat{\bm \beta}^{0}/\lVert 
\hat{\bm \beta}^{0}\rVert)$.  For this numerical experiment, we observe that $h=5$ and $h=0.1$ gave the best empirical performance across the selected $h$ values. The Adaptive choice of $h$ performed comparably to the other choices of $h$ and not far from the best performing $h$ values.}
\label{tab:doolr_h}
\scalebox{0.8}{
\begin{tabular}{ccccccccc}
\toprule
$n$ & Measure 
& $h = 0.003$ 
& $h = 0.02$ 
& $h = 0.10$
& $h = 0.50$
& $h = 1.00$
& $h = 5.00$
& Adaptive \\
\midrule
\multirow{2}{*}{200} 
  & TPR & 0.884(0.197) & 0.911(0.178) & 0.906(0.131) & 0.843(0.236) & 0.844(0.278) & 0.905(0.203) & 0.875(0.154) \\\\
  & PPV & 0.319(0.105) & 0.306(0.120) & 0.376(0.079) & 0.440(0.134) & 0.336(0.087) & 0.339(0.093) & 0.435(0.105) \\\\
\midrule
\multirow{2}{*}{400} 
  & TPR & 0.887(0.204) & 0.888(0.205) & 0.927(0.118) & 0.863(0.254) & 0.933(0.205) & 0.959(0.100) & 0.894(0.177) \\\\
  & PPV & 0.347(0.091) & 0.348(0.080) & 0.386(0.071) & 0.431(0.117) & 0.354(0.058) & 0.356(0.060) & 0.430(0.100) \\\\
\midrule
\multirow{2}{*}{800} 
  & TPR & 0.900(0.199) & 0.906(0.193) & 0.959(0.076) & 0.858(0.306) & 0.911(0.264) & 0.964(0.120) & 0.923(0.188) \\\\
  & PPV & 0.368(0.064) & 0.372(0.068) & 0.430(0.067) & 0.402(0.108) & 0.345(0.074) & 0.363(0.058) & 0.425(0.087) \\\\
\bottomrule
\end{tabular}}
\end{table}

\begin{table}[h]
\caption{Data generation is under the linear decision rule with no contamination: mean and standard deviation (in parentheses) of estimates of TPR and PPV for varying values of the PPV constraint $(\alpha)$.}
\label{tab:linear}
\centering
\scalebox{0.8}
{\begin{tabular}{cccccccccc}
\toprule
\multicolumn{1}{c}{$\alpha$} & \multicolumn{1}{c}{$n$} & \multicolumn{1}{c}{ } & \multicolumn{1}{c}{\textbf{Measure}}& \multicolumn{4}{c}{\textbf{Methods}} \\
\cmidrule(rl){5-8} 
\multicolumn{4}{c}{} &{Standard} & {Plug-in Logistic} & {Plug-in SL}  & {DOOLR} \\
\midrule
0.030 & 1250 &  & TPR & 0.993(0.004) & 0.993(0.007) & 0.992(0.010) & 0.984(0.030) \\\\
      &      &  & PPV & 0.033(0.003) & 0.030(0.002) & 0.030(0.012) & 0.032(0.004) \\\\
      & 2500 &  & TPR & 0.994(0.002) & 0.995(0.003) & 0.995(0.004) & 0.990(0.011) \\\\
      &      &  & PPV & 0.032(0.002) & 0.030(0.001) & 0.029(0.007) & 0.032(0.005) \\\\
\midrule
0.040 & 1250 &  & TPR & 0.988(0.005) & 0.988(0.009) & 0.988(0.011) & 0.975(0.036) \\\\
      &      &  & PPV & 0.042(0.003) & 0.040(0.002) & 0.039(0.011) & 0.044(0.009) \\\\
      & 2500 &  & TPR & 0.991(0.003) & 0.991(0.005) & 0.991(0.006) & 0.983(0.017) \\\\
      &      &  & PPV & 0.042(0.002) & 0.040(0.002) & 0.040(0.005) & 0.043(0.006) \\\\
\midrule
0.045 & 1250 &  & TPR & 0.984(0.007) & 0.984(0.010) & 0.984(0.013) & 0.963(0.049) \\\\
      &      &  & PPV & 0.047(0.003) & 0.045(0.003) & 0.043(0.011) & 0.050(0.013) \\\\
      & 2500 &  & TPR & 0.986(0.004) & 0.986(0.007) & 0.986(0.007) & 0.976(0.019) \\\\
      &      &  & PPV & 0.046(0.002) & 0.045(0.002) & 0.045(0.005) & 0.049(0.007) \\\\
\bottomrule
\end{tabular}}
\end{table}

\begin{table}[t]
\caption{Data generation is under the linear decision rule with contamination: mean and standard deviation (in parentheses) of estimates of TPR and PPV for varying values of the PPV constraint $(\alpha)$.}
\label{tab:linearCon}
\centering
\scalebox{0.8}
{\begin{tabular}{cccccccccc}
\toprule
\multicolumn{1}{c}{$\alpha$} & \multicolumn{1}{c}{$n$} & \multicolumn{1}{c}{ } & \multicolumn{1}{c}{\textbf{Measure}}& \multicolumn{4}{c}{\textbf{Methods}} \\
\cmidrule(rl){5-8} 
\multicolumn{4}{c}{} &{Standard} & {Plug-in Logistic} & {Plug-in SL}  & {DOOLR} \\
\midrule
0.030 & 1250 &  & TPR & 0.625(0.308) & 0.690(0.288) & 0.986(0.019) & 0.968(0.068) \\\\
      &      &  & PPV & 0.038(0.019) & 0.028(0.007) & 0.030(0.002) & 0.033(0.006) \\\\
      & 2500 &  & TPR & 0.677(0.287) & 0.738(0.288) & 0.989(0.008) & 0.979(0.033) \\\\
      &      &  & PPV & 0.043(0.021) & 0.029(0.005) & 0.030(0.001) & 0.032(0.005) \\\\
\midrule
0.040 & 1250 &  & TPR & 0.519(0.295) & 0.608(0.306) & 0.973(0.022) & 0.939(0.091) \\\\
      &      &  & PPV & 0.044(0.025) & 0.033(0.012) & 0.040(0.003) & 0.044(0.009) \\\\
      & 2500 &  & TPR & 0.606(0.271) & 0.664(0.291) & 0.977(0.012) & 0.962(0.041) \\\\
      &      &  & PPV & 0.051(0.022) & 0.037(0.008) & 0.040(0.002) & 0.043(0.006) \\\\
\midrule
0.045 & 1250 &  & TPR & 0.510(0.293) & 0.567(0.298) & 0.965(0.024) & 0.926(0.084) \\\\
      &      &  & PPV & 0.044(0.024) & 0.035(0.014) & 0.045(0.004) & 0.048(0.010) \\\\
      & 2500 &  & TPR & 0.591(0.269) & 0.626(0.289) & 0.970(0.012) & 0.948(0.048) \\\\
      &      &  & PPV & 0.052(0.022) & 0.040(0.011) & 0.045(0.003) & 0.049(0.007) \\\\
\bottomrule
\end{tabular}}
\end{table}

\begin{table}[t]
\caption{Data generation is under the non-linear decision rule: mean and standard deviation (in parentheses) of estimates of TPR and PPV for varying values of the PPV constraint $(\alpha)$.}
\label{tab:nonlinear}
\centering
\scalebox{0.8}
{\begin{tabular}{cccccccccc}
\toprule
\multicolumn{1}{c}{$\alpha$} & \multicolumn{1}{c}{$n$} & \multicolumn{1}{c}{ } & \multicolumn{1}{c}{\textbf{Measure}}& \multicolumn{4}{c}{\textbf{Methods}} \\
\cmidrule(rl){5-8} 
\multicolumn{4}{c}{} &{Standard} & {Plug-in Logistic} & {Plug-in SL}  & {DOOLR} \\
\midrule
0.030 & 1250 &  & TPR & 0.740(0.282) & 0.930(0.145) & 0.993(0.023) & 0.932(0.120) \\\\
      &      &  & PPV & 0.041(0.014) & 0.031(0.004) & 0.030(0.003) & 0.031(0.006) \\\\
      & 2500 &  & TPR & 0.641(0.338) & 0.959(0.102) & 0.996(0.004) & 0.952(0.081) \\\\
      &      &  & PPV & 0.041(0.013) & 0.030(0.002) & 0.030(0.001) & 0.031(0.003) \\\\
\midrule
0.040 & 1250 &  & TPR & 0.682(0.261) & 0.796(0.236) & 0.983(0.027) & 0.862(0.153) \\\\
      &      &  & PPV & 0.047(0.009) & 0.041(0.007) & 0.040(0.003) & 0.038(0.011) \\\\
      & 2500 &  & TPR & 0.613(0.301) & 0.839(0.201) & 0.988(0.010) & 0.879(0.124) \\\\
      &      &  & PPV & 0.047(0.009) & 0.041(0.004) & 0.040(0.002) & 0.039(0.009) \\\\
\midrule
0.045 & 1250 &  & TPR & 0.634(0.250) & 0.735(0.236) & 0.965(0.067) & 0.856(0.161) \\\\
      &      &  & PPV & 0.049(0.009) & 0.045(0.009) & 0.046(0.005) & 0.038(0.016) \\\\
      & 2500 &  & TPR & 0.584(0.259) & 0.755(0.224) & 0.981(0.017) & 0.846(0.141) \\\\
      &      &  & PPV & 0.050(0.008) & 0.045(0.005) & 0.045(0.002) & 0.041(0.013) \\\\
\bottomrule
\end{tabular}}
\end{table}

\begin{table}[t]
\caption{Data generation is under the piece-wise linear decision rule: mean and standard deviation (in parentheses) of estimates of TPR and PPV for varying values of the PPV constraint $(\alpha)$.}
\label{tab:piece-wise}
\centering
\scalebox{0.8}
{\begin{tabular}{cccccccccc}
\toprule
\multicolumn{1}{c}{$\alpha$} & \multicolumn{1}{c}{$n$} & \multicolumn{1}{c}{ } & \multicolumn{1}{c}{\textbf{Measure}}& \multicolumn{4}{c}{\textbf{Methods}} \\
\cmidrule(rl){5-8} 
\multicolumn{4}{c}{} &{Standard} & {Plug-in Logistic} & {Plug-in SL}  & {DOOLR} \\
\midrule
0.030 & 1250 &  & TPR & 0.457(0.115) & 0.514(0.094) & 0.919(0.117) & 0.499(0.101) \\\\
      &      &  & PPV & 0.035(0.012) & 0.031(0.007) & 0.031(0.012) & 0.044(0.032) \\\\
      & 2500 &  & TPR & 0.508(0.104) & 0.561(0.075) & 0.961(0.057) & 0.545(0.088) \\\\
      &      &  & PPV & 0.036(0.012) & 0.031(0.005) & 0.030(0.003) & 0.044(0.030) \\\\
\midrule
0.040 & 1250 &  & TPR & 0.373(0.131) & 0.433(0.128) & 0.909(0.130) & 0.458(0.106) \\\\
      &      &  & PPV & 0.044(0.016) & 0.041(0.012) & 0.044(0.018) & 0.069(0.049) \\\\
      & 2500 &  & TPR & 0.423(0.106) & 0.475(0.100) & 0.946(0.061) & 0.512(0.088) \\\\
      &      &  & PPV & 0.048(0.014) & 0.042(0.009) & 0.041(0.006) & 0.062(0.035) \\\\
\midrule
0.045 & 1250 &  & TPR & 0.368(0.134) & 0.423(0.128) & 0.869(0.128) & 0.458(0.120) \\\\
      &      &  & PPV & 0.050(0.015) & 0.046(0.014) & 0.049(0.019) & 0.077(0.048) \\\\
      & 2500 &  & TPR & 0.355(0.123) & 0.415(0.120) & 0.922(0.071) & 0.503(0.093) \\\\
      &      &  & PPV & 0.051(0.015) & 0.046(0.009) & 0.045(0.005) & 0.069(0.038) \\\\
\bottomrule
\end{tabular}}
\end{table}

\clearpage

 \section{Appendix D}
\subsection{Proof of theorems}
\noindent The proof of the Theorem 1 is based on Lemma 2 from  \cite{meisner2021combining} and M-estimation theory \citep{van2000asymptotic}. Lemma 2 is stated and proved in \cite{meisner2021combining} hence we are not going to prove but just state it here.

We first define the following notations, $\widehat{\text{FPR}}^{\Phi}(\beta_0,\bm \beta_1)\coloneqq n_0^{-1}\sum_{j=1}^{n_{0}}\Phi(\frac{\bm{\beta_0 + \bm X_{0j}\beta}_1}{h})$, $\widehat{\text{TPR}}^{\Phi}(\beta_0,\bm \beta_1)\coloneqq n_1^{-1}\sum_{i=1}^{n_{1}}\Phi(\frac{\bm{\beta_0 + \bm X_{1i}\beta}_1}{h})$, $\text{FPR}^{\Phi}(\beta_0,\bm \beta_1)\coloneqq \text{Pr}\left(\Phi(\frac{\bm{\beta_0 + \bm X\beta}_1}{h})|D=0 \right)$, and  $\text{TPR}^{\Phi}(\beta_0,\bm \beta_1)\coloneqq \text{Pr}\left(\Phi(\frac{\bm{\beta_0 + \bm X\beta}_1}{h})|D=1 \right)$.

\begin{lemma}[Lemma 2 from \cite{meisner2021combining}]  \label{lemma1}
    Under conditions (1)–(5), we have that
\begin{align*}
   \underset{(\beta_0,\bm \beta_1) \in \Omega}{\text{sup }} |\widehat{\text{FPR}}^{\Phi}(\beta_0,\bm \beta_1) - \text{FPR}(\beta_0,\bm \beta_1) |\rightarrow 0,  \quad  \underset{(\beta_0,\bm \beta_1) \in \Omega}{\text{sup }} |\widehat{\text{TPR}}^{\Phi} (\beta_0,\bm{\beta_1}) - \text{TPR} (\beta_0,\bm{\beta_1}) |\rightarrow 0
\end{align*}

    almost surely as $n\rightarrow \infty$, where $\Omega=\{(\beta_0,\bm \beta_1)\in  \mathbb{R} \times \mathbb{R}^{p}:||\bm \beta_1||=1\}$.
\end{lemma}

\subsection*{Proof of theorem 1(a)}
We want to show that $TPR(\hat{\bm \beta}^{\Phi}) \rightarrow TPR(\bm \beta^*)$ in probability.
\begin{proof}
   Define $\bm \beta\coloneqq (\beta_{0},\bm \beta_{1})$, $\widehat{\text{TPR}}(\bm \beta)\coloneqq n_1^{-1}\sum_{i=1}^{n_{1}}\mathbbm{1} \{\bm X_{1i} \bm \beta > 0\}$, and $\widehat{\text{FPR}}(\bm \beta)\coloneqq n_0^{-1}\sum_{j=1}^{n_{0}}\mathbbm{1} \{\bm X_{0j} \bm \beta > 0\}$, then
   \begin{align*}
       |\text{TPR}(\hat{\bm \beta}^{\Phi}) -\text{TPR}(\bm \beta^*)| & \le   |\text{TPR}(\hat{\bm \beta}^{\Phi}) - \widehat{\text{TPR}}(\hat{\bm \beta}^{\Phi}) + \widehat{\text{TPR}}(\hat{\bm \beta}^{\Phi}) -\text{TPR}(\bm \beta^*)| \\
       &  \le   \underset{\bm \beta \in \Omega}{\text{ sup }}|\text{TPR}(\bm \beta) - \widehat{\text{TPR}}(\bm \beta)| + |\widehat{\text{TPR}}(\hat{\bm \beta}^{\Phi})- \text{TPR}^{\Phi}(\hat{\bm \beta}^{\Phi})|  \\ 
       & + |\text{TPR}^{\Phi}(\hat{\bm \beta}^{\Phi}) -\text{TPR}(\hat{\bm \beta}^{\Phi})| + |\text{TPR}(\hat{\bm \beta}^{\Phi}) -\text{TPR}(\bm \beta^*)|
   \end{align*}
 \noindent $\underset{\bm \beta \in \Omega}{\text{ sup }}|\text{TPR}(\bm \beta) - \widehat{\text{TPR}}(\bm \beta)| $ converges to 0 by Lemma~\ref{lemma1}, $ |\text{TPR}^{\Phi}(\hat{\bm \beta}^{\Phi}) -\text{TPR}(\hat{\bm \beta}^{\Phi})|$ converges to 0 as $h \rightarrow 0$, $|\widehat{\text{TPR}}(\hat{\bm \beta}^{\Phi})- \text{TPR}^{\Phi}(\hat{\bm \beta}^{\Phi})| $ converges to 0 by Glivenko-Cantelli theorem, and $|\text{TPR}(\hat{\bm \beta}^{\Phi}) -\text{TPR}(\bm \beta^*) | $ converges to 0 by the continuous mapping theorem (we assumed $\text{TPR}(\bm \beta)$ is Lipschitz continuous).
\end{proof}

\subsection*{Proof of theorem 1(b)}
Here we want to show that $\text{lim sup}_n PPV(\hat{\bm \beta}^{\Phi})\ge \alpha$ in probability

\begin{proof}
    First, we define the continuous function $f(x,y)\coloneqq \frac{p_1x}{p_1x +(1-p_1)y}$ and express $PPV(\hat{\bm \beta} ^{\Phi})\coloneqq \frac{p_1 TPF(\hat{\bm \beta}^{\Phi})}{p_1 TPF(\hat{\bm \beta}^{\Phi}) +(1-p_1)FPF(\hat{\bm \beta}^{\Phi})}$ where 
    $\hat{\bm \beta}^{\Phi}=(\hat{\beta}_0^{\Phi},\hat{\bm \beta}_1^{\Phi})$, $\bm \beta=(\beta_0,\bm \beta_1)$,  and $p_1$ is prevalence which we assumed to be fixed. Then we have that
    \begin{align*}
        PPV(\hat{\bm \beta}^{\Phi})&=|\widehat{PPV}^{\Phi}(\hat{\bm \beta}^{\Phi}) + \{PPV(\hat{\bm \beta}^{\Phi})-\widehat{PPV}^{\Phi}(\hat{\bm \beta}^{\Phi}) \}| \\
        &\ge |\widehat{PPV}^{\Phi}(\hat{\bm \beta}^{\Phi})| - |PPV(\hat{\bm \beta}^{\Phi})-\widehat{PPV}^{\Phi}(\hat{\bm \beta}^{\Phi})| \\
        & \ge \alpha - \underset{\bm \beta \in \Omega}{\text{sup }}|PPV(\bm \beta)-\widehat{PPV}^{\Phi}(\bm \beta)| \\
    \end{align*}
    From Lemma~\ref{lemma1}, we have that  $\underset{\bm \beta \in \Omega}{\text{sup }}|PPV(\bm \beta)-\widehat{PPV}^{\Phi}(\bm \beta)|$ converges to zero almost surely hence
    \begin{align*}
         \text{Pr}\{\underset{n}{\text{lim inf }} \text{PPV}(\hat{\bm \beta}) \ge \alpha \}  \ge  \text{Pr}\{\underset{n}{\text{lim inf  }} \underset{\bm \beta \in \Omega}{\text{sup }}|PPV(\bm \beta)-\widehat{PPV}^{\Phi}(\bm \beta)|=0 \}=1.
    \end{align*}
\end{proof}

\subsection*{Proof of theorem 1(c)} We want to show that $\hat{\bm \beta}^{\Phi} \rightarrow \bm \beta^*$ in probability. 
\begin{proof}
    
To show the above, we will employ Theorem 5.7 from \cite{van2000asymptotic}. We first assume $\bm \beta^*$ is a well-separated point of maximum of $TPR(\bm \beta)$ and show that the following conditions:  
\begin{enumerate}
    \item $\underset{\bm \beta \in \Omega}{\text{sup }} | \widehat{TPR}^{\Phi}(\bm \beta) - TPR(\bm \beta)| \overset{p}{\rightarrow}0$ \quad (\text{convergence in probability})
    \item $\widehat{TPR}^{\Phi}(\hat{\bm \beta}^{\Phi}) \ge \widehat{TPR}^{\Phi}(\bm \beta^*)-o_P(1)$ hold.
\end{enumerate}

From Lemma~\ref{lemma1}, we have $\underset{\bm \beta \in \Omega}{\text{sup }} | \widehat{TPR}^{\Phi}(\bm \beta) - TPR(\bm \beta)| \overset{p}{\rightarrow}0$ hence $1)$  hold. To show 2), we  note that $\hat{\bm \beta}^{\Phi}$ is a near maximizer of $\widehat{TPR}^{\Phi}(\bm \beta)$, that is $ \widehat{TPR}^{\Phi}(\hat{\bm \beta}^{\Phi}) \ge \underset{\bm \beta \in \Omega}{\text{sup }} \widehat{TPR}^{\Phi}(\bm \beta) -o_{P}(1).$ We then have that
\begin{align*}
    \widehat{TPR}^{\Phi}(\hat{\bm \beta}^{\Phi}) & \ge \underset{\bm \beta \in \Omega}{\text{sup }} \widehat{TPR}^{\Phi}(\bm \beta) -o_{P}(1) \\
    \widehat{TPR}^{\Phi}(\hat{\bm \beta}^{\Phi}) & \ge \underset{\bm \beta \in \Omega}{\text{sup }} \widehat{TPR}^{\Phi}(\bm \beta) + \widehat{TPR}^{\Phi}(\bm \beta^*)-   \widehat{TPR}^{\Phi}(\bm \beta^*) -o_{P}(1) \\
    \widehat{TPR}^{\Phi}(\hat{\bm \beta}^{\Phi}) &  \ge \widehat{TPR}^{\Phi}(\bm \beta^*) + \underset{\bm \beta \in \Omega}{\text{sup }} \widehat{TPR}^{\Phi}(\bm \beta) - TPR(\bm\beta^*) + TPR(\bm\beta^*) -   \widehat{TPR}^{\Phi}(\bm \beta^*) - o_{P}(1) \\
     \widehat{TPR}^{\Phi}(\hat{\bm \beta}^{\Phi}) &  \ge \widehat{TPR}^{\Phi}(\bm \beta^*) + \underset{\bm \beta \in \Omega}{\text{sup }} \widehat{TPR}^{\Phi}(\bm \beta) - \underset{\bm \beta \in \Omega }{\text {sup }}TPR(\bm\beta) + TPR(\bm\beta^*) -   \widehat{TPR}^{\Phi}(\bm \beta^*) - o_{P}(1)
\end{align*}
We therefore have that $\underset{\bm \beta \in \Omega}{\text{sup }} \widehat{TPR}^{\Phi}(\bm \beta) - \underset{\bm \beta \in \Omega }{\text {sup }}TPR(\bm\beta)\le  \underset{\bm \beta \in \Omega}{\text{sup }}| \widehat{TPR}^{\Phi}(\bm \beta) -TPR(\bm\beta) | \rightarrow 0$ and $TPR(\bm\beta^*) -   \widehat{TPR}^{\Phi}(\bm \beta^*) \le \underset{\bm \beta \in \Omega}{\text{sup }}| TPR(\bm\beta) -   \widehat{TPR}^{\Phi}(\bm \beta) | \rightarrow 0$ almost surely from Lemma~\ref{lemma1} hence we have $\widehat{TPR}^{\Phi}(\hat{\bm \beta}^{\Phi}) \ge \widehat{TPR}^{\Phi}(\bm \beta^*)  - o_{P}(1)$. Therefore we can conclude that $\hat{\bm \beta}^{\Phi} \overset{p}{\rightarrow}\bm \beta.$
\end{proof}

\end{document}